\documentclass[journal,draftcls,onecolumn,12pt,twoside]{IEEEtranTCOM}
\renewcommand{\baselinestretch}{1.5}
\usepackage{subfigure}
\usepackage{moreverb}
\usepackage{epsfig}
\usepackage{amsmath,amssymb,amsthm,mathrsfs,amsfonts,dsfont}
\usepackage{adjustbox,lipsum}
\usepackage{algorithm,algorithmic}
\usepackage{amsfonts}
\usepackage{epsfig}
\usepackage{amssymb}
\usepackage{amsmath}
\usepackage{amsthm}
\usepackage{subfigure}
\usepackage{multirow}
\usepackage{rotating}
\usepackage{graphicx}
\usepackage{tabularx}
\usepackage{array}
\usepackage{anyfontsize}
\usepackage{color,soul}
\usepackage{graphicx,dblfloatfix}
\usepackage{epstopdf}
\usepackage{blindtext}
\usepackage{amsmath}
\usepackage{amsthm,amssymb,amsmath,bm}
\usepackage{subfigure}
\usepackage{amsfonts}
\usepackage{epsfig}
\usepackage{amssymb}
\usepackage{amsmath}
\usepackage{cite}
\usepackage{graphicx}
\usepackage{fancyhdr}
\usepackage{subfigure}
\usepackage[subfigure]{tocloft}
\usepackage[font={small}]{caption}
\usepackage{subfigure}
\usepackage{tabularx}
\usepackage{tcolorbox}
\usepackage{cite}
\DeclareMathOperator*{\argmax}{arg\,max}

\newtheorem{theorem}{Theorem}
\newtheorem{lemm}{Lemma}
\newtheorem{assumption}{Assumption}
\newtheorem{Definition}{Definition}
\newtheorem{rem}{Remark}

\renewcommand{\baselinestretch}{1.6}

\begin{document}

\title{ Power Minimization for Age of Information Constrained Dynamic Control in Wireless Sensor Networks
	% via Stochastic Optimization
%%	Lyapunov
%	%A Stochastic Optimization Approach
%%	and Energy Efficiency In  CSMA/CA Based Wireless Sensor Networks %with Packet Management
%\author{  \IEEEauthorblockN{Mohammad Moltafet and  Markus Leinonen  }
%	\IEEEauthorblockA{Centre for Wireless Communications -- Radio Technologies\\
%		University of Oulu, Finland\\
%		e-mail: $\{$mohammad.moltafet, markus.leinonen$\}$@oulu.fi}	
%	%\textbf{Technical Area:  C.4. Optimization, C.2. IoT  }
%	\and
%	\IEEEauthorblockN{  Marian Codreanu and Nikolaos Pappas}
%	\IEEEauthorblockA{Department of Science and Technology\\
%		Link\"{o}ping University, Sweden\\
%	
%		e-mail: $\{$marian.codreanu, nikolaos.pappas$\}$@liu.se}
%} 
%}	
	\thanks{ This research has been financially supported by the Infotech Oulu, the Academy of Finland (grant 323698), and Academy of Finland 6Genesis Flagship (grant 318927). M. Codreanu would like to acknowledge the support of the European Union's Horizon 2020 research and innovation programme under the Marie Sk\l{}odowska-Curie Grant Agreement No. 793402 (COMPRESS NETS). The work of M. Leinonen has also been financially supported in part by the Academy of Finland (grant 319485). 
	M. Moltafet would like to acknowledge the support of Finnish Foundation for Technology Promotion,  HPY Research Foundation, Riitta ja Jorma J. Takanen Foundation, and Nokia Foundation.}
\thanks{
	\IEEEauthorrefmark{1}Mohammad Moltafet and Markus Leinonen are  with the Centre
	for Wireless Communications--Radio Technologies, University of Oulu,
	90014 Oulu, Finland (e-mail: mohammad.moltafet@oulu.fi; markus.leinonen@oulu.fi).}
	\thanks{
	\IEEEauthorrefmark{2}Marian Codreanu and Nikolaos Pappas are with Department of Science and Technology, Link\"{o}ping University, Sweden (e-mail: marian.codreanu@liu.se; nikolaos.pappas@liu.se).}
%\thanks{	
%	  The work of M. Moltafet, M. Leinonen and M. Codreanu has been financially supported by the Infotech Oulu, the Academy of Finland (grant 323698), and Academy of Finland 6Genesis Flagship (grant 318927). M. Codreanu would like to acknowledge the support of the European Union's Horizon 2020 research and innovation programme under the Marie Sk\l{}odowska-Curie Grant Agreement No. 793402 (COMPRESS NETS). The work of M. Leinonen has also been financially supported in part by the Academy of Finland (grant 319485). M. Moltafet would like to acknowledge the support of Finnish Foundation for Technology Promotion and  HPY Research Foundation.
%}
\thanks{	
	Preliminary results of this paper were presented in \cite{moltafet2019power}.
}
%\author{ 
%	Mohammad~Moltafet,  \textit{Student Member, IEEE}, Markus~Leinonen, \textit{Member, IEEE}, Marian~Codreanu, \textit{Member, IEEE}, and Nikolaos~Pappas, \textit{Member, IEEE}
%}
\author{ 
	Mohammad~Moltafet\IEEEauthorrefmark{1}, Markus~Leinonen\IEEEauthorrefmark{1}, Marian~Codreanu\IEEEauthorrefmark{2}, and Nikolaos~Pappas\IEEEauthorrefmark{2}
}
%\vspace{-.4\baselineskip}
}

%\author{ 
%	Mohammad~Moltafet\IEEEauthorrefmark{1},~\IEEEmembership{Student Member,~IEEE,} Markus~Leinonen\IEEEauthorrefmark{1},~\IEEEmembership{Member,~IEEE,}  Marian~Codreanu\IEEEauthorrefmark{2},~\IEEEmembership{Member,~IEEE,} and Nikolaos~Pappas\IEEEauthorrefmark{2},~\IEEEmembership{Member,~IEEE}
%}

\maketitle
\vspace{-2.9\baselineskip}
\begin{abstract} 
	We consider a system where multiple sensors communicate timely information about various random processes to a sink. The sensors share orthogonal sub-channels to transmit such information in the form of status update packets. A central controller can control the sampling actions of the sensors to trade-off between the transmit power consumption and information freshness which is quantified by the Age of Information (AoI). We jointly optimize the sampling action of each sensor, the transmit power allocation, and the sub-channel assignment to minimize the average total transmit power of all sensors subject to a maximum average AoI constraint for each sensor. To solve the problem, we develop a dynamic control algorithm using the Lyapunov drift-plus-penalty method and provide optimality analysis of the algorithm. According to the Lyapunov drift-plus-penalty method, to solve the main problem we need to solve an optimization problem in each time slot which is a mixed integer non-convex optimization problem. We propose a low-complexity sub-optimal solution for this per-slot optimization problem that provides near-optimal performance and we evaluate the computational complexity of the solution. Numerical results illustrate the performance of the proposed dynamic control algorithm and the performance of the sub-optimal solution for the per-slot optimization problems versus the different parameters of the system. The results show that the proposed dynamic control algorithm achieves more than $60~\%$ saving in the average total transmit power compared to a baseline policy.
	
\end{abstract}	
%%	\section{Introduction}\label{Introduction}
%%	
%%
%%\begin{itemize}
%%	\item h
%%\end{itemize}
%
%
\allowdisplaybreaks

\allowdisplaybreaks

\section{Introduction}

Freshness of the status information of various physical processes collected by multiple sensors is a key performance enabler in many applications of wireless sensor networks (WSNs) \cite{8187436,8469047,Sunbook2019}, e.g., surveillance in smart home systems and drone control. The Age of Information (AoI) was introduced as a destination centric metric that characterizes the freshness in status update systems \cite{6195689,6310931,5984917,Sunbook2019,7415972}.
A status update packet  of each sensor contains a time stamp representing  the time when the sample was generated and the measured value of the monitored process. Due to wireless channel access, channel errors, and fading etc., communicating a status update packet through the network experiences a random delay. If at a time instant $t$, the most recently received status update packet
contains the time stamp $U(t)$, the AoI is defined
as the random process $\Delta(t)=t-U(t)$. In other words,  the AoI of each sensor  is the time elapsed since the last received status update packet was generated at the sensor. 
%One of the commonly used metric to
%evaluate the AoI is the average AoI . 
In this work, we focus on the average AoI which is a commonly used metric to evaluate the AoI \cite{8469047,6195689,6310931,5984917,8486307,chen2019optimal,moltafet2019age,8901143, moltafet2020ISIT,8006703,8772205,bhat2019throughput,8123937,8606155,8734015,8933047,9163054}.
% In this work, we
%focus on the average AoI.

Besides the requirement of high information freshness, low energy consumption is vital for maintaining a status update WSN operational. Namely, the wireless sensors are typically battery limited, thus may be infeasible to recharge or replace batteries during the operation. The main contributors to the sensors' energy resources are the wireless access \cite{Raghunathan-02}, and also, the sensing/sampling part \cite{Anastasi-etal-09}. Consequently, it is crucial to minimize the amount of information (e.g., the number of data packets) that must be communicated from each sensor to the sink to meet the application requirements. This engenders the need for joint optimization of the information freshness, sensors' sampling policies, and radio resource allocation (transmit power, bandwidth etc.) for designing energy-efficient status update WSNs.

%\textcolor{blue}{Energy consumption is an other critical parameter in operation of the WSNs. The power source of the sensors is typically  a battery with a limited power budget. Moreover, it could be impossible  to recharge the battery. 
%On the other hand, the sensor network has to have a lifetime long enough to meet the application requirements. 
%Therefore, considering energy efficiency in the WSNs in order to evaluate the system performance is very important. \cite{777291,5339143,8255543}.}

%In this paper, we consider a WSN and assume a slotted communication in which the sensors can control the sampling action. 
\subsection{Contributions}
We consider a WSN consisting of a set of sensors and a sink that is interested in time-sensitive information from the sensors. We minimize the average total transmit power of sensors by jointly optimizing  the sampling action, the transmit power allocation, and the sub-channel assignment under the constraint on the maximum average AoI of each sensor. To solve the proposed problem, we develop  a dynamic control algorithm using the Lyapunov drift-plus-penalty method. In addition, we provide optimality analysis of the proposed dynamic control  algorithm. According to the Lyapunov drift-plus-penalty method,  to solve the main  problem we need to solve an optimization problem in each time slot which is a mixed integer non-convex optimization problem. We propose a low-complexity sub-optimal solution for this per-slot optimization problem that provides near-optimal performance and evaluate the computational complexity of the solution. 
%and the optimality of the sub-optimal solution in the numerical results section.  
%for the proposed optimization problem. 
%According to the Lyapunov drift-plus-penalty method, in each time slot, we need to solve an optimization problem which is a mixed integer non-convex optimization problem. We propose a sub-optimal solution to solve the optimization problem in each slot.
Numerical results show the performance of the proposed dynamic control algorithm  in terms of  transmit power consumption and AoI of the sensors versus  different   system parameters. In addition, they show that  the 
sub-optimal  solution for the per-slot optimization problems is near-optimal. 
%According to the Lyapunov drift-plus-penalty method,  to solve the main  problem we need to solve an optimization problem in each time slot which is a mixed integer non-convex optimization problem. We propose a sub-optimal solution to solve the optimization problem in each slot and evaluate the computational complexity and the optimality of the sub-optimal solution in the numerical results section.

%This work extends our preliminary work in \cite{moltafet2019power}. In \cite{moltafet2019power}, we did not provide the optimality analysis of the proposed dynamic control algorithm. Moreover, in \cite{moltafet2019power}, we  applied the exhaustive search method to solve the optimization problem that needs to be solved in each slot and consequently we considered a small setup with limited numerical results. In this work, in addition to providing optimality analysis of the proposed dynamic control  algorithm and proposing a sub-optimal and fast solution to solve the optimization problem in each slot, we provide more detailed algorithm description and more extensive simulation results.

\subsection{Related Work}

Since the introduction of the AoI, it has been under extensive study in various communication setups. For example,  AoI under various queueing models were studied in \cite{6195689,6310931,6875100,8006504,7282742,8469047,moltafet2020ISIT,moltafet2020sourceawareage,moltafet2019age,9119460}; AoI in energy harvesting based WSNs were investigated in \cite{8606155,7283009,7308962,8006703,8822722}; and  AoI under various channel access models were studied in \cite{8006544,5984917,Kosta2018AgeOI,8901143,9007478}.

There are only a few works in which optimization of radio resource allocation, scheduling, and sensor sampling action has been studied. 
  The authors of \cite{8006703} considered  an energy harvesting sensor and derived  the optimal threshold in terms of  remaining energy to trigger a new sample to minimize the AoI.  In  
 \cite{8772205}, the authors considered a status update system in which the updates of different sensors are generated with fixed rate and proposed a power control algorithm to minimize the average AoI. 
  The work in \cite{bhat2019throughput} considered a single user fading channel system and studied long-term average throughput maximization subject to average AoI and power constraints.
  The authors of  \cite{8123937} considered  an energy harvesting sensor and minimized the average AoI by determining the optimal status update policy.  
% The authors of \cite{chen2019optimal} considered  two source nodes generating heterogeneous traffic with different power supplies  and  studied the peak-age-optimal  status update scheduling. 
 %The authors of \cite{8606155} considered a  wireless power
% transfer powered sensor network  and studied performance of the system  in terms of the average  AoI.
% The author of \cite{8901143} analyzed the worst case average AoI and peak AoI for each sensor in a carrier sense multiple access  based WSN. The authors of \cite{moltafet2019age} studied the average AoI and peak AoI in  multi-source queueing models.  
 %The most related work to this paper is \cite{8734015}. 
The work in \cite{8734015} considered a WSN in which sensors share one  unreliable sub-channel  in each slot. 
%They assumed that the required power for each packet transmission is fixed.
They minimized the  expected weighted sum average AoI of the network
% with minimum average throughput requirement  constraint for each sensor 
by determining the transmission scheduling policy.
In \cite{8933047}, the authors considered a system where a base station serves multiple traffic streams arriving according to a stochastic process and the packets of different streams are enqueued in  separate  queues.  They minimized the  expected weighted sum AoI of the network
% with minimum average throughput requirement  constraint for each sensor 
by determining the transmission scheduling policy.
In \cite{9163054}, the authors considered a multi-user system in which users share one  unreliable sub-channel  in each slot. They proposed an optimization problem to minimize the cost of sampling and transmitting status updates in the system under an average AoI constraint for each user in the system. They solved the problem by the  Lyapunov drift-plus-penalty method.

While the prior works contain different combinations of AoI-aware sampling, scheduling, and power optimization, to the best of our knowledge, this is the first work that proposes the joint optimization over the listed three WSN parameters: transmit power allocation, sub-channel assignment, and sampling action.

  %To the best of our knowledge, joint optimization of the transmit power allocation, sub-channel assignment, and sampling action with constrained AoI  has not been studied earlier.

\subsection{Organization}
The rest of this paper is organized as follows. The system model and problem formulation  are  presented in Section \ref{System Model and Problem Formulation}. The Lyapunov drift-plus-penalty method to solve the proposed problem is presented in Section \ref{Solution Algorithm}. The optimality analysis of the proposed dynamic control algorithm to solve the main  problem is provided in Section \ref{Optimality Analysis of the Proposed Solution}. The proposed  sub-optimal solution  for the mixed integer non-convex optimization problem in each slot is presented in Section \ref{A Suboptimal Solution to Solve Problem}.
 Numerical results  are presented in Section \ref{Numerical Results}. Finally, the concluding remarks are expressed in Section \ref{Conclusion}.
\section{ System Model and Problem Formulation}\label{System Model and Problem Formulation}
In this section, we present the considered system model and the problem formulation.
\subsection{System Model}
We consider a WSN consisting of a set $\mathcal{K}$ of $K$ sensors and one sink,  as depicted  in Fig.~\ref{model}. The sink is interested in time-sensitive information from the sensors which measure  physical phenomena.  We assume slotted communication with normalized slots ${t\in\{0,1,\dots\}}$, where in each slot, the sensors share a set $\mathcal{N}$ of $N$ orthogonal sub-channels  with bandwidth $W$ Hz per sub-channel. We consider that  a central controller controls the sampling process of  sensors in such a way that it decides whether each sensor takes a sample or not
at the beginning of each slot $t$.

 We assume that the perfect channel state information of all sub-channels is available at the central controller at the beginning of each slot.
%, and we do not have knowledge for future slots. 
Let $h_{k,n}(t)$ denote the channel coefficient from sensor $k$ to the sink over sub-channel $n$ in slot $t$. We assume that 	
%\begin{assumption}\label{asum1}\normalfont
	$h_{k,n}(t)$ is a stationary process and it is 
	independent and identically distributed (i.i.d) over slots.
%\end{assumption}

Let $\rho_{k,n}(t)$ denote the sub-channel assignment at time slot $t$ as ${\rho_{k,n}(t)\in \{0,1\},}$${ \forall k \in\mathcal{K}, n \in\mathcal{N}}$, where   $\rho_{k,n}(t)=1$ indicates that sub-channel $n$ is assigned to sensor $k$ at  time slot $t$, and  $\rho_{k,n}(t)=0$ otherwise. To ensure that at any given time slot $t$, each sub-channel can be assigned to at most one sensor, the following constraint is used
\begin{align}\label{mn001}
\sum_{k\in\mathcal{K}}\rho_{k,n}(t)\le 1, \forall n\in \mathcal{N}, t.
\end{align}

\begin{figure}
	\centering
	\includegraphics[scale=.44]{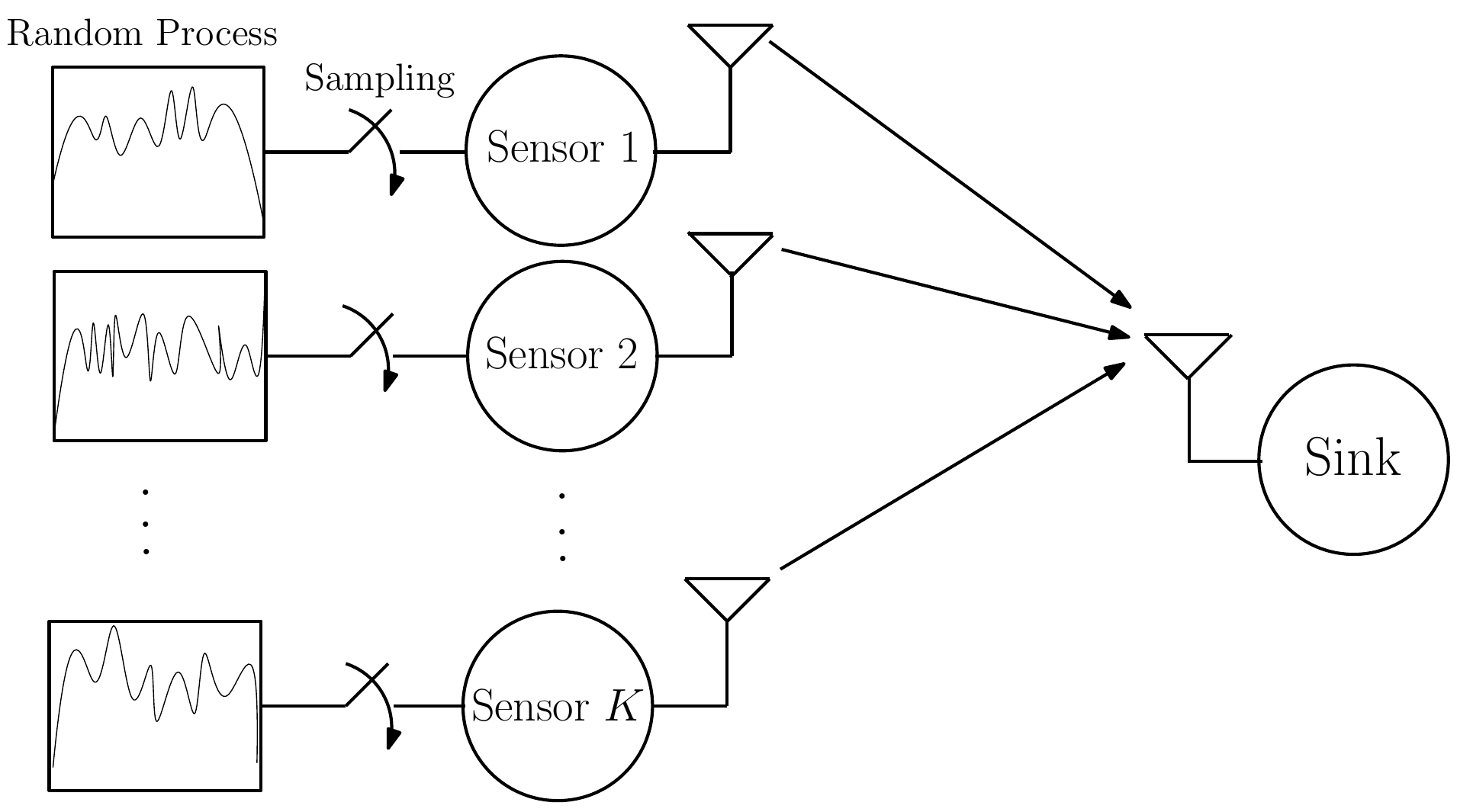}
	%\vspace{-5mm}
	\caption{A WSN consisting of  $K$ sensors and one sink that receives time-sensitive information from the sensors.}
	\vspace{-13mm}
	\label{model}
\end{figure}

Let $p_{k,n}(t)$ denote the transmit power of sensor $k$ over sub-channel $n$ in slot $t$.
%For a compact notation, we introduce the matrices  $\mathbf{P}(t)=[p_{k,n}(t)], \forall k \in\mathcal{K}, n \in\mathcal{N}$, and  $\boldsymbol{\rho}(t)=[\rho_{k,n}(t)], \forall k \in\mathcal{K}, n \in\mathcal{N}$, respectively. 
Then,  the signal-to-noise ratio  with respect to sensor $k$ over sub-channel $n$ in slot $t$ is given by 
\begin{align}
\gamma_{k,n}(t)=\dfrac{p_{k,n}(t)|h_{k,n}(t)|^2}{WN_0},
\end{align}
where $N_0$ is the noise power spectral density. The achievable  rate for sensor $k$ over  sub-channel $n$ in slot $ t $  is given by
\begin{align}
r_{k,n}(t)=W\log_2\left(1+\gamma_{k,n}(t)\right).
\end{align}
% Here, we assume that fractions of packets can be sent.
The achievable data rate of sensor $k$ in slot $t$ is the sum of the achievable data rates over all the assigned sub-channels at  slot $t$,   expressed  as 
\begin{align}\nonumber
R_{k}(t)=\sum_{n\in\mathcal{N}}\rho_{k,n}(t)r_{k,n}(t).
\end{align}

   Let  $b_k(t)$ denote the sampling action of  sensor $k$ at time slot $t$  as $b_k(t)\in \{0,1\}, \forall k \in\mathcal{K}$, where   $b_k(t)=1$ indicates that sensor  $k$  takes a sample at the beginning of time slot $t$, and  $b_k(t)=0$ otherwise. We assume that  sampling time (i.e.,  the time needed to acquire a sample) is negligible. 
   %For a compact notation, we introduce the matrix  $\mathbf{B}=[b_{k}(t)], \forall k \in\mathcal{K}, t\in\mathcal{T}$.
%number of generated packets of sensor $k$ on time slot $t$, which is a function of random wireless channel 
We consider that the central controller decides that sensor $k$ takes a sample at the beginning of slot $t$ only if there are enough resources  to guarantee that the sample is successfully transmitted during the same slot $t$. 
Thus, if sensor $k$ takes a sample at the beginning of slot $t$ (i.e., $b_k(t)=1$), the sample will be transmitted during the same slot $t$ successfully.   To this end, we use the following constraint
\begin{align}
R_{k}(t)= \eta b_{k}(t), \forall k\in \mathcal{K}, t,
\end{align}
 where  $\eta$ is the  size of each status update packet (in bits).
This constraint ensures that when sensor $k$ takes a sample at the beginning of slot $t$ (i.e., $b_k(t)=1$), the achievable rate for sensor $k$ in slot $t$  is $R_{k}(t)= \eta$, guaranteeing that the sample is transmitted during the slot.

Let $\delta_k(t)$ denote the AoI of sensor $k$ at the beginning of slot $t$. 
%We consider that the value of AoI of each sensor is updated at the beginning of the slot that follows a transmission of the sensor\footnote{This assumption simplifies the problem while captures the features }.
 If sensor $k$ takes a sample at the beginning of slot $t$ (i.e., $b_k(t)=1$),  the AoI  at the beginning of slot $t+1$ drops to one, and
  %is $\delta_k(t+1)=1$,
   otherwise (i.e., $b_k(t)=0$), the AoI  increases by one. Thus, the evolution of $\delta_k(t)$ is  characterized as 
\begin{align}\label{AoI1}
\delta_k(t+1)&=\begin{cases}
1,&\!\!\!\!\text{if} \,\,b_k(t)=1;\\
\delta_k(t)+1,&\!\!\!\!\text{otherwise}.
\end{cases}
\end{align}
The evolution of the AoI of  sensor $k$ is illustrated in Fig. \ref{AoI}.

\begin{figure}
	\centering
	\includegraphics[scale=1]{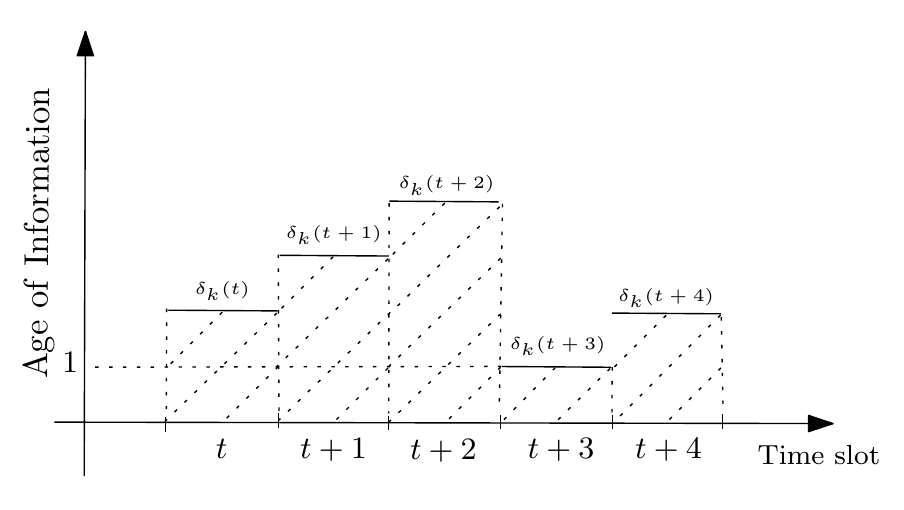}%{AoI02.pdf}
	\caption{The evolution of the AoI of sensor $k$. Without status updates, the AoI increases by one unit during each slot; a status update received during slot $ t+2 $ caused the AoI to drop to one at the beginning of slot $ t+3 $. }
	 		\vspace{-10mm}
	\label{AoI}
\end{figure}
  
%
%
%
%The time average AoI of sensor $k$ is calculated as the  area under the AoI curve, normalized by the observation interval. As it can be seen, during slot $t$, the area under the AoI curve of sensor $k$ is 
%%calculated as a sum of the areas of a triangle and a parallelogram.  The area of the triangle is equal to $1/2$ and the area of the  parallelogram is 
%equal to $\delta_k(t)$. {Therefore, the time average AoI of sensor $k$ is calculated as 
%\begin{align}\label{mnbhg010}
%\bar\Delta_k&=
%%\lim_{t\to \infty}\dfrac{1}{ t}\sum_{t=0}^{T-1}\mathbb{E}[1/2+\delta_k(t)]\\&\nonumber=
%\lim_{T\to \infty}\dfrac{1}{T}\sum_{t=0}^{T-1}\delta_k(t).
%\end{align}}
%
%{To make the calculations tractable, we use a commonly used approach that instead of the time average AoI in \eqref{mnbhg010}, 
%we consider the time average of expectation of the AoI \cite{8486307,Sunbook2019,8933047,8734015},
%
Following a commonly used approach \cite{8486307,Sunbook2019,8933047,8734015}, we define the average AoI of sensor $k$ as  the time average of the expected value of the AoI given as
\begin{align}\label{mnbhg01}
\Delta_k&=
%\lim_{t\to \infty}\dfrac{1}{ t}\sum_{t=0}^{T-1}\mathbb{E}[1/2+\delta_k(t)]\\&\nonumber=
\lim_{T\to \infty}\dfrac{1}{T}\sum_{t=0}^{T-1}{\mathbb{E}}[\delta_k(t)],
\end{align}
where the expectation is with respect to the  random wireless channel states and control actions made in reaction to the channel states\footnote{In this work, all expectations are taken with respect to the randomness of the wireless channel states and control actions made in reaction to the channel states.}. Without loss of generality,  we consider that the initial value of the AoI of all sensors is   $\delta_k(0)=0, \,\,\forall k\in \mathcal{K}$. 
\subsection{Problem Formulation}
%In this section, we present the proposed optimization problem. 
%The main aim of the  optimization problem
Our objective  is to minimize the  average total transmit power of all sensors by jointly optimizing  the sampling action, the transmit power allocation, and the sub-channel assignment in each slot subject to  the maximum average AoI    constraint for each sensor.
Thus, the problem is formulated as follows 
%\textcolor{red}{Marian: Regarding the spacing between equation lines, I used the Boyd's template for all multi-line equations. However, here since we need numbering for constraints (e.g., (7a), (7b), ...) the template does not work. I will find a way to reduce the  space between lines. }
% Variables are transmit power in each sensor, subcarrier assignment and arrival rate of sensors. 
  \begin{subequations}\label{eqo1}
 	\begin{align}\label{eq8a}
 	&\text{minimize} \,\,\lim_{T\to \infty}\dfrac{1}{T}\textstyle\sum_{t=0}^{T-1}\textstyle\sum_{k\in \mathcal{K}}\textstyle\sum_{n\in\mathcal{N}}\mathbb{E}[p_{k,n}(t)]\\&\label{eqo5}
 	\text{subject\,\,to}\hspace{.15cm}
\lim_{T\to \infty}\dfrac{1}{T}\textstyle\sum_{t=0}^{T-1}{\mathbb{E}}[\delta_k(t)]\le \Delta^{\text{max}}_k,~\forall k\in \mathcal{K}
%\\&\label{eqo3}
% 	\hspace{1.8cm}\sum_{n\in\mathcal{N}}p_{k,n}(t)\le P_k^{\text{max}},~ k\in \mathcal{K}, \forall t
 	\\&\label{eq1o5}
 	\hspace{1.85cm}\textstyle\sum_{n\in\mathcal{N}}\rho_{k,n}(t)W\log_2\left(1+\dfrac{p_{k,n}(t)|h_{k,n}(t)|^2}{WN_0}\right)= \eta b_{k}(t),~\forall  k\in \mathcal{K},~ t
 	\\&\label{eq8o2} 
 	\hspace{1.85cm}\textstyle\sum_{k\in\mathcal{K}}\rho_{k,n}(t)\le 1,~ \forall n\in \mathcal{N}, t
 	\\&\label{eq1o5000}
 	\hspace{1.85cm}p_{k,n}(t)\ge 0,~\forall  k\in \mathcal{K}, n\in \mathcal{N},~ t
 	%0\le p_{k,n}(t)\le p^{\text{max}},\forall k\in \mathcal{K}, n\in \mathcal{N}, t\in \mathcal{T},
 	 \\&\label{eq1o50}
 	\hspace{1.85cm}\rho_{k,n}(t)\in\{0,1\},~\forall  k\in \mathcal{K}, n\in \mathcal{N},~ t
 	 \\&\label{eq1o51}
 	\hspace{1.85cm} b_{k}(t)\in\{0,1\},~\forall  k\in \mathcal{K},~ t,
 	\end{align}
 \end{subequations}
with variables $\{p_{k,n}(t),\rho_{k,n}(t)\}_{k \in\mathcal{K}, n \in\mathcal{N}}$ and $\{b_{k}(t)\}_{k \in\mathcal{K}}$ for all ${t\in\{0,1,\dots\}}$, where $\Delta^{\text{max}}_k$ is the maximum acceptable average AoI of sensor $k$. The constraints of  problem \eqref{eqo1} are as follows.
The  inequality  \eqref{eqo5} is the maximum acceptable average  AoI constraint for each sensor; 
%the inequality 
%$p^{\text{max}}$ is the maximum power that can be transmitted over each sub-channel at a time slot; 
%\eqref{eqo3}  constrains the power of each sensor with respect to the maximum budget $P_k^{\text{max}}$;
%\footnote{In this paper, we assume that the value of $P_k^{\text{max}}$ is large enough so that we do not have infallibility problem.};
% 
the equality  \eqref{eq1o5} 
 ensures that each sample is transmitted during one slot; the inequality \eqref{eq8o2} constrains  that each sub-channel can be assigned to at most one sensor in each slot; \eqref{eq1o5000}, \eqref{eq1o50}, and \eqref{eq1o51} represent the feasible values for the transmit power, sub-channel assignment, and sampling policy variables, respectively. 
 
Problem \eqref{eqo1} is a mixed integer non-convex problem where the constraints and the objective function both contain  averages over the optimization variables. In the next section,  a dynamic control algorithm is proposed to solve   problem \eqref{eqo1}.

Prior to that, we introduce the definitions of feasibility of  problem  \eqref{eqo1}, {channel-only} policies, and the Slater's condition  for  problem  \eqref{eqo1} which are needed in our optimality analysis in  Section \ref{Optimality Analysis of the Proposed Solution}.

%[Feasibility of problem \eqref{eqo1}]
\begin{Definition}\normalfont
	Problem \eqref{eqo1} is \emph{feasible} if there exists a policy that satisfies  constraints \eqref{eqo5}--\eqref{eq1o51}  \cite[Sect.~4.3]{neely2010stochastic}.
\end{Definition}
\begin{Definition}	\normalfont
%[{Channel-only} policies]
	The \textit{channel-only} policies are a class of policies that   make decisions for sampling action, power allocation and sub-channel assignment  of each sensor independently every slot $ t $ based only on the observed channel state  \cite[Sect.~3.1]{neely2010stochastic}. 
	%\textcolor{red} {Do we need to explicitly say that these policies need the channel statistics?}
\end{Definition}
%  Next, we present the main assumptions in the system.       	

%        	Let  $p^*_{k,n}(t)$ and $\delta_k^*(t)$  denote power allocation and the value of AoI determined  by a \textit{channel-only} policy, respectively. Then, because of Assumption \ref{asum1} we have \cite[Sect.~ 4.3]{neely2010stochastic}
%        	\begin{align}
%        	&\sum_{k\in \mathcal{K}}\sum_{n\in\mathcal{N}}\mathbb{E}\left[p^*_{k,n}(t)\right]=\lim_{T\to \infty}\dfrac{1}{T}\sum_{t=0}^{T-1}\sum_{k\in \mathcal{K}}\sum_{n\in\mathcal{N}}p^*_{k,n}(t),\\&\nonumber
%        	\mathbb{E}[\delta^*_k(t)]=\lim_{T\to \infty}\dfrac{1}{T}\sum_{t=0}^{T-1}\delta_k^*(t).
%        	\end{align}

\begin{assumption}\label{asum3} \normalfont
We assume that problem \eqref{eqo1} satisfies	Slater's condition \cite[Sect.~4.3]{neely2010stochastic}, i.e., there are values $\epsilon>0$, $\hat G(\epsilon)\ge0$, and a {channel-only}  policy that satisfy in each slot
	\begin{align}\label{slaterassump1}
	&\sum_{k\in \mathcal{K}}\sum_{n\in\mathcal{N}}\mathbb{E}\left[\hat p_{k,n}(t)\right]=\hat G(\epsilon),\\&\label{slaterassump2}
	\mathbb{E}[\hat\delta_k(t)]+\epsilon\le\Delta^{\text{max}}_k,~ \forall k\in\mathcal{K},
	\end{align}
	where $\hat p_{k,n}(t)$ and $\hat \delta_k(t)$  denote the allocated power to sensor $k$ over sub-channel $n$ in slot $ t $ and the value of the AoI of sensor $k$  in slot $ t $  determined by the  {channel-only} policy, respectively. 
%	According to this assumption, we assume that there is a {channel-only} policy so that the virtual queues are strongly stable. 
	%average AoI constraints in \eqref{eqo5} can be  tightened  by $\epsilon$ and the constraints remain still feasible.
\end{assumption}    
% \textcolor{red} {Marian: It would be nice to give some examples of necessary conditions under which this assumption is automatically satisfied. 
 %	Mohammad: I think about it.}

\section{Dynamic Control Algorithm}\label{Solution Algorithm}
In this section,  we  develop a dynamic control algorithm to solve problem \eqref{eqo1}.
% and then present the algorithm summary.
%\subsection{Lyapunov Drift-Plus-Penalty Optimization}
To this end, we use the Lyapunov drift-plus-penalty method \cite{neely2010stochastic}, \cite{stochast009om}. According to the drift-plus-penalty method, the  average AoI constraints \eqref{eqo5} are enforced by transforming them into queue stability constraints. For each inequality  constraint \eqref{eqo5}, a virtual queue is associated in such a way that the stability of these virtual queues implies the feasibility of the average AoI constraint \eqref{eqo5}.

%To use the drift-plus penalty method, let us first rewrite constraint \eqref{eqo5} using \eqref{mnbhg01} as 
%\begin{align}\label{consta1}
%\lim_{T\to \infty}\dfrac{1}{T}\sum_{t=0}^{T-1}\mathbb{E}[\delta_k(t)]\le \Delta^{\text{max}}_k,~\forall k \in\mathcal{K}.
%\end{align}
%In the following we define the Lyapunov function and its drift  which are needed to define the queue stability problem.
Let $\{Q_k(t)\}_{k\in \mathcal{K}}$ denote the virtual queues associated with AoI constraint  \eqref{eqo5}. The virtual queues are updated in each time slot as 
\begin{align}\label{consta2}
Q_k(t+1)=\max\left[Q_k(t)-\Delta^{\text{max}}_k,0\right]+\delta_k(t+1),~ \forall k\in \mathcal{K}.
\end{align}
 Here, we use the notion of strong stability; the virtual queues are strongly stable if \cite[Ch. 2]{neely2010stochastic} 
\begin{align}\label{consta3}
\lim_{T\to \infty}\dfrac{1}{T}\sum_{t=0}^{T-1}\mathbb{E}[Q_k(t)]<\infty,~\forall k\in \mathcal{K}.
\end{align}
According to \eqref{consta3}, a queue is strongly stable if its  average mean backlog is finite. Note that the strong stability of the virtual queues in \eqref{consta2} implies that the  average AoI constraint \eqref{eqo5} is satisfied.  Next, we introduce the Lyapunov function and its drift  which are needed to define the queue stability problem. 
% Let $\bold{Q}(t)$ denote a vector containing all the virtual queues, i.e, ${\bold{Q}(t)=[Q_1(t),Q_2(t),\dots,Q_K(t)]}$, 
% 

 Let $\mathcal{S}(t)=\{Q_k(t),\delta_k(t)\}_{k\in \mathcal{K}}$ denote the network state at slot $t$,  and $\bold{Q}(t)$ denote a vector containing all the virtual queues, i.e., ${\bold{Q}(t)=[Q_1(t),Q_2(t),\ldots,Q_K(t)]\in\mathbb{R}^{1\times K}}$.
 Then, a quadratic Lyapunov  function $L(\bold{Q}(t))$  is defined by \cite[Ch. 3]{neely2010stochastic} 
 \begin{align}\label{consta4}
L(\bold{Q}(t))=\dfrac{1}{2}\sum_{k\in\mathcal{K}}Q^2_k(t).
 \end{align}
 The Lyapunov function measures the network congestion:  if the Lyapunov function is small, then all the queues are small, and if the Lyapunov function is large, then at least one queue is large. Therefore, by minimizing the expected change of the Lyapunov function from one slot to the next slot, queues $\{Q_k(t)\}_{k\in \mathcal{K}}$  can be stabilized \cite[Ch. 4]{neely2010stochastic}.
 \begin{Definition}\normalfont
 The conditional Lyapunov drift $ \alpha(\mathcal{S}(t)) $  is defined as the expected change in the Lyapunov function over one slot given that the current network state in slot $t$ is $\mathcal{S}(t)$. Thus, $ \alpha(\mathcal{S}(t)) $  is given by
  \begin{align}\label{consta5}
 \alpha(\mathcal{S}(t))=\mathbb{E}\left[L\left(\bold{Q}(t+1)\right)-L\left(\bold{Q}(t)\right)\mid\mathcal{S}(t)\right].
 \end{align}
 \end{Definition}
According to the drift-plus-penalty minimization method, a control policy that minimizes the objective function of problem \eqref{eqo1} with constraints \eqref{eqo5}--\eqref{eq1o51} is obtained 
by solving the following problem 
%minimizing the drift-plus-penalty function in 
%each slot $t$ 
\cite[Ch. 3]{neely2010stochastic}
%, which is given by
% \begin{align}\label{consta6}
%\alpha(\mathcal{S}(t))+V\sum_{k\in \mathcal{K}}\sum_{n\in\mathcal{N}}\mathbb{E}[p_{k,n}(t)|\mathcal{S}(t)],
%\end{align}
%  subject to the following constraints
%  \begin{equation}
%\begin{array}{ll}
%\mbox{minimize}   & 
  \begin{subequations}\label{eqo2}
 	\begin{align}\label{eq8a2}
 		&\text{minimize}\,\,\,\alpha(\mathcal{S}(t))+V\textstyle\sum_{k\in \mathcal{K}}\textstyle\sum_{n\in\mathcal{N}}\mathbb{E}[p_{k,n}(t)\mid\mathcal{S}(t)]\\&\label{eq8o201}
 	\text{subject\,\,to}\hspace{.15cm}\eqref{eq1o5}-\eqref{eq1o51} 
 	\end{align}
 \end{subequations}
with variables $\{p_{k,n}(t),\rho_{k,n}(t)\}_{k \in\mathcal{K}, n \in\mathcal{N}}$ and $\{b_{k}(t)\}_{k \in\mathcal{K}}$, where a parameter $V\ge0$ is used to adjust the emphasis on the objective function (i.e., power minimization). Therefore, by varying $V$, a desired trade-off between the  sizes of the queue backlogs  and the objective function value 
 can be obtained. 

Since 
%\eqref{consta6} is a function of $\mathbb{E}[L(\bold{Q}(t+1))|\mathcal{S}(t)]$, 
minimizing the objective function \eqref{eq8a2} is intractable,  we  minimize an upper bound of \eqref{eq8a2} \cite[Ch. 3]{neely2010stochastic}.
  To find an upper bound for \eqref{eq8a2}, we find an upper bound for the conditional Lyapunov drift $ \alpha(\mathcal{S}(t)) $. To this end,  we use the following inequality in which  for any $\hat{A}\ge0$, $\tilde{A}\ge0$, and $\bar{A}\ge0$ we have \cite[Ch. 3]{neely2010stochastic}
 \begin{align}\label{wrf01}
 \left(\max\left[\hat{A}-\tilde{A},0\right]+\bar{A}\right)^2\le\hat{A}^2+\tilde{A}^2+\bar{A}^2+2\hat{A}(\bar{A}-\tilde{A}).
 \end{align}
 
 By applying \eqref{wrf01} to \eqref{consta2}, an upper bound  for $ Q^2_k(t+1)$  is given as 
 \begin{align}\label{021mb4}
& Q^2_k(t+1)\le Q^2_k(t)+\left(\Delta^{\text{max}}_k\right)^2+\delta^2_k(t+1)+2Q_k(t) \left(\delta_k(t+1)-\Delta^{\text{max}}_k\right),~ \forall k\in \mathcal{K}.
 \end{align}
 
 By applying \eqref{021mb4} to the conditional Lyapunov drift $ \alpha(\mathcal{S}(t)) $, we obtain an upper bound to \eqref{consta5} as
 %an upper bound for \eqref{consta6} is given as
\begin{equation}
\begin{array}{ll}
%\begingroup
%\begin{align}
  \alpha(\mathcal{S}(t))\le \dfrac{1}{2}\mathbb{E}\bigg[\sum_{k\in\mathcal{K}}\bigg((\Delta^{\text{max}}_k)^2+\delta^2_k(t+1)+2Q_k(t)\big(\delta_k(t+1)-\Delta^{\text{max}}_k\big)\bigg)\bigg|\mathcal{S}(t)\bigg]\\=\dfrac{1}{2}\sum_{k\in\mathcal{K}}\bigg((\Delta^{\text{max}}_k)^2+\mathbb{E}[\delta^2_k(t+1)\mid\mathcal{S}(t)]+2Q_k(t)\big(\mathbb{E}[\delta_k(t+1)\mid\mathcal{S}(t)]-\Delta^{\text{max}}_k\big)\bigg).
 % \end{align}
 % \endgroup
      \end{array}
      \label{021mb40}
   \end{equation}
   
   To characterize the upper bound in \eqref{021mb40}, we need to determine  ${\mathbb{E}[\delta_k(t+1)\mid\mathcal{S}(t)]}$ and ${\mathbb{E}[\delta^2_k(t+1)\mid\mathcal{S}(t)]}$ in \eqref{021mb40}. To this end, by using the  evolution of the AoI in \eqref{AoI1},  ${\delta_k(t+1)}$ and ${\delta^2_k(t+1)}$ are calculated as 
     \begin{equation}
     \begin{array}{ll}
     \delta_k(t+1)=b_k(t)+\left(1-b_k(t)\right)(\delta_k(t)+1),~\forall k\in \mathcal{K} \\
     \delta^2_k(t+1)=b_k(t)+\left(1-b_k(t)\right)(\delta_k(t)+1)^2,~\forall  k\in \mathcal{K}.
      \end{array}
       \label{021mb401}
          \end{equation}
        By using the expressions in \eqref{021mb401}, $\mathbb{E}[\delta_k(t+1)\mid\mathcal{S}(t)]$ and  $\mathbb{E}[\delta^2_k(t+1)\mid\mathcal{S}(t)]$ in \eqref{021mb40} are given as 
        \begin{equation}
        \begin{array}{ll}
        \mathbb{E}[\delta_k(t+1)\mid\mathcal{S}(t)]=\mathbb{E}[b_k(t)\mid\mathcal{S}(t)]+(1-\mathbb{E}[b_k(t)\mid\mathcal{S}(t)])(\delta_k(t)+1),~\forall k\in \mathcal{K} \\
                \mathbb{E}[\delta^2_k(t+1)\mid\mathcal{S}(t)]=\mathbb{E}[b_k(t)\mid\mathcal{S}(t)]+(1-\mathbb{E}[b_k(t)\mid\mathcal{S}(t)])(\delta_k(t)+1)^2,~\forall k\in \mathcal{K}.
     \end{array}
\label{021mb4010}
\end{equation}
        
        By substituting  \eqref{021mb4010} into the right hand side of  \eqref{021mb40}, and adding the term ${V\sum_{k\in \mathcal{K}}\sum_{n\in\mathcal{N}}\mathbb{E}[p_{k,n}(t)\mid\mathcal{S}(t)]}$ to both sides of \eqref{021mb40}, the upper bound for \eqref{eq8a2} is given as 
  \begin{equation}
\begin{array}{ll}
        \alpha(\mathcal{S}(t))+V\sum_{k\in \mathcal{K}}\sum_{n\in\mathcal{N}}\mathbb{E}[p_{k,n}(t)\mid\mathcal{S}(t)]\le\\ \mathbb{E}\left[V\sum_{k\in \mathcal{K}}\sum_{n\in\mathcal{N}}p_{k,n}(t)+ \dfrac{1}{2}\sum_{k\in\mathcal{K}}b_k(t)\big(1-(\delta_k(t)+1)^2 -{{2Q_k(t)}}\delta_k(t)\big)\bigg|\mathcal{S}(t)\right] \\+\dfrac{1}{2}\sum_{k\in\mathcal{K}}\bigg((\Delta^{\text{max}}_k)^2+(\delta_k(t)+1)^2+2Q_k(t)(\delta_k(t)+1)-2Q_k(t)\Delta^{\text{max}}_k\bigg).
             \end{array}
\label{uppervn}
\end{equation}
        Having defined the upper bound \eqref{uppervn}, instead of minimizing  \eqref{eq8a2}, we minimize \eqref{uppervn} subject to the constraints \eqref{eq1o5}--\eqref{eq1o51} with variables $\{p_{k,n}(t),\rho_{k,n}(t)\}_{k \in\mathcal{K}, n \in\mathcal{N}}$ and $\{b_{k}(t)\}_{k \in\mathcal{K}}$. Given that we observe the channel states $\{h_{k,n}(t)\}_{k\in \mathcal{K},n\in \mathcal{N}}$ at the beginning of each slot, we use the  approach of \textit{opportunistically minimizing an expectation} to solve the optimization problem. 
        According to this approach, \eqref{uppervn} is minimized by ignoring the expectations in each slot.
          Note that the  approach of opportunistically minimizing an expectation provides  the optimal control policy \cite[Sect.~1.8]{neely2010stochastic}.

        \begin{algorithm}[h]
        	{   \caption{Proposed dynamic control algorithm for problem \eqref{eqo1}}
        		\label{table-1}
        		Step 1. \textbf{Initialization}: set ${t=0}$, set $V$, and  initialize $\{Q_k(0)=0,\delta_k(0)=0\}_{k\in \mathcal{K}}$
        		\\
        		\textbf{for}  each time slot $t$ \textbf{do}
        		\\
        		Step 2. \textbf{Sampling action, transmit power, and sub-channel assignment}: obtain $\{p_{k,n}(t),\rho_{k,n}(t)\}_{k \in\mathcal{K}, n \in\mathcal{N}}$ and $\{b_{k}(t)\}_{k \in\mathcal{K}}$ by solving the following optimization problem  
        		%by using the solution algorithm presented in Section \ref{A Suboptimal Solution to Solve Problem}
        		%\begin{subequations}
                  \begin{equation}
                  \begin{array}{ll}
        		\mbox{minimize}   & V\sum_{k\in \mathcal{K}}\sum_{n\in\mathcal{N}}p_{k,n}(t)+\dfrac{1}{2}\sum_{k\in\mathcal{K}} %(\Delta^{\text{max}}_k)^2+(\delta_k(t)+1)^2+(2Q_k(t)-1)(\delta_k(t)+1)+ \\&\label{eqia} \hspace{2cm}
        		b_k(t)\left[1-(\delta_k(t)+1)^2-2Q_k(t)\delta_k(t)\right]\\
        		\mbox{subject to} & 
        		\eqref{eq1o5}-\eqref{eq1o51},
        		  \end{array}
                   \label{eqoi1}
        		\end{equation}
        		%\end{subequations}
        		$~~~$with variables $\{p_{k,n}(t),\rho_{k,n}(t)\}_{k \in\mathcal{K}, n \in\mathcal{N}}$ and $\{b_{k}(t)\}_{k \in\mathcal{K}}$
        		\\
        		Step 3. \textbf{Queue update}:  update $\{Q_k(t+1),\delta_k(t+1)\}_{k\in \mathcal{K}}$ using \eqref{consta2} and  \eqref{021mb401}
        		\\
        		$~~~~~~~~~~~~$Set $t=t+1$, and go to Step 2
        		\\
        		\textbf{end for}
        	}
        \end{algorithm}
        The main steps of the proposed dynamic control  algorithm are summarized in  Algorithm~\ref{table-1}.  
        The controller observes the channel states $\{h_{k,n}(t)\}_{k\in \mathcal{K},n\in \mathcal{N}}$ and network state  $\mathcal{S}(t)$ at the beginning of each time slot $t$. Then following the  approach of opportunistically minimizing an expectation, it takes a control action to minimize \eqref{eqoi1} subject to the constraints \eqref{eq1o5}--\eqref{eq1o51} in Step~2.  Note that the objective function of \eqref{eqoi1} follows from \eqref{uppervn} because i) 
       the variables of the optimization problem are $\{p_{k,n}(t),\rho_{k,n}(t)\}_{k \in\mathcal{K}, n \in\mathcal{N}}$ and $\{b_{k}(t)\}_{k \in\mathcal{K}}$ and thus,  we neglected the second term of the upper bound \eqref{uppervn} as it does not depend on the optimization variables and ii) 
                the opportunistically minimizing an expectation approach minimizes \eqref{uppervn} by ignoring the expectations in each slot. 
         In Step~3, according to the solution of \eqref{eqoi1}, the virtual queue and AoI of each sensor are updated by using  \eqref{consta2} and \eqref{021mb401}, respectively.
         
          It is worth to note that the optimization problem \eqref{eqoi1} is a mixed integer non-convex optimization problem containing both integer (i.e., sub-channel assignment and sampling action) and continuous (i.e., power allocation) variables. One way to find the optimal solution of  problem \eqref{eqoi1} is to use an exhaustive search method.
          %To find the optimal solution an exhaustive search method can be used.
           However,  it suffers from high computational complexity which increases exponentially with the number of variables in the system. Therefore, in Section~\ref{A Suboptimal Solution to Solve Problem}, we propose a sub-optimal solution for the  optimization problem \eqref{eqoi1}. Before that, we analyze the optimality of the proposed dynamic algorithm which is carried out in the next section.

      \section{Optimality Analysis of the Proposed Solution}\label{Optimality Analysis of the Proposed Solution}  
        	In this section, we study the performance of the proposed Lyapunov drift-plus-penalty method (i.e., Algorithm~\ref{table-1}) used to solve  problem \eqref{eqo1}. In particular, the main result of our analysis will be stated in Theorem~\ref{theo1}, which characterizes the trade-off between the optimality of the objective function (i.e., the average total transmit power) and average backlogs of the virtual queues in \eqref{consta2}.
        	%Before presenting Theorem \ref{theo1}, we introduce the definitions and assumptions that are needed in our analysis.
        	%to study the performance of the Lyapunov drift-plus-penalty method.
        	%prove Theorem \ref{theo1}.  
        	
        	We first note an important property of the AoI evolution under the proposed dynamic control algorithm: the Lyapunov drift-plus-penalty Algorithm~\ref{table-1} ensures that the AoI of sensors are \textit{bounded}, i.e., there is a constant $\delta^{\text{max}}< \infty$ such that $\delta_k(t)\le \delta^{\text{max}},\forall k\in\mathcal{K}, t$. 
        	Recall that the main goal of Algorithm~\ref{table-1} is to minimize the objective function of \eqref{eqoi1} in each slot. The objective function of \eqref{eqoi1} can be written as a form 
        	$\sum_{k=1}^{K}f(b_k(t),\delta_k(t)),$ where $f(b_k(t),\delta_k(t))=V\sum_{n\in\mathcal{N}}p_{k,n}(t)-1/2b_k(t)\left[-1+(\delta_k(t)+1)^2+2Q_k(t)\delta_k(t)\right]$. The maximum value of $f(b_k(t),\delta_k(t))$ for each sensor $k$ is zero and it is achieved when the sensor does not take a sample in slot $t$, i.e., $b_k(t)=0$. However, if sensor $k$ does not take a sample, its AoI increases by one after each slot, and thus, after some slots the term $1/2\left[-1+(\delta_k(t)+1)^2+2Q_k(t)\delta_k(t)\right]$ of $f(b_k(t),\delta_k(t))$ becomes  greater than the first term $V\sum_{n\in\mathcal{N}}p_{k,n}(t)$;
        	%(i.e., multiplication of the required power to transmit a packet and $ V $) 
        	 in this case, it is optimal for sensor $k$ to take a sample since it makes  $f(b_k(t),\delta_k(t))$ negative.
        	%which results in a negative $f(b_k(t),\delta_k(t))$, and consequently,  sensor $k$ takes a sample. 
        	 Thus, we conclude that each sensor takes a sample in a finite number of time slots, and this implies that there is a constant $\delta^{\text{max}}< \infty$ such that $\delta_k(t)\le \delta^{\text{max}},\forall k\in\mathcal{K}, t$.

         	Next, we present Lemma \ref{theo01} which shows that if  problem \eqref{eqo1} is feasible, we can get arbitrarily close to the optimal solution by {channel-only} policies. This lemma  is used to prove Theorem \ref{theo1}.
             	\begin{lemm}\label{theo01}\normalfont
             		Under the assumption that  each  channel is a stationary process and i.i.d over slots, if   problem \eqref{eqo1} is feasible, then for any $\nu>0$ there is a {channel-only} policy that satisfies 
        			\begin{align}\label{cha00}
        		&\sum_{k\in \mathcal{K}}\sum_{n\in\mathcal{N}}\mathbb{E}\left[p^*_{k,n}(t)\right]\le G^{\text{opt}}+\nu,\\&\label{cha01}
        		\mathbb{E}[\delta^*_k(t)]-\nu\le \Delta^{\text{max}}_k,~ \forall k\in\mathcal{K},
        		\end{align}
        		where $ G^{\text{opt}} $ denotes the optimal value of the average total transmit power (i.e., the optimal value of the objective function of  problem \eqref{eqo1}), and $p^*_{k,n}(t)$ and $\delta_k^*(t)$  denote the allocated power to sensor $k$ over sub-channel $n$ and the value of the AoI of sensor $k$ determined by the  {channel-only} policy, respectively.
             	\end{lemm}
        		\begin{proof}
        			See proof of Theorem 4.5 in \cite[Appendix~ 4.A]{neely2010stochastic}.
        		\end{proof}
        	
        Next, we present Theorem \ref{theo1} which characterizes a trade-off between the optimality of the objective function of  problem \eqref{eqo1} and the average backlogs of the virtual queues in the system.% under Algorithm \ref{table-1}.
            \begin{theorem}\label{theo1}\normalfont
        		Suppose that  problem \eqref{eqo1} is feasible  and $L\left(\bold{Q}(0)\right)< \infty$. Then,	for  any values of parameter $V>0$,  Algorithm~\ref{table-1}  satisfies the average AoI constraints in \eqref{eqo5}. Further, let       $\bar p_{k,n}(t)$ denote the allocated power to sensor $k$ over sub-channel $n$ in slot $ t $ as determined  by Algorithm~\ref{table-1}, and  $\bar Q_k(t)$  denote the virtual queue of sensor $k$ in slot $ t $ as determined  by Algorithm~\ref{table-1}. Then,   we have  the following upper bounds for the average total transmit power and  the average backlogs of the virtual queues in the system
        			\begin{align}\label{optimalobject}
        			\lim_{T\to \infty}\dfrac{1}{T}\sum_{t=0}^{T-1}\sum_{k\in \mathcal{K}}\sum_{n\in\mathcal{N}}\mathbb{E}[\bar p_{k,n}(t)]\le \dfrac{B}{V}+G^{\text{opt}},
        			\end{align}
        			\begin{align}\label{optimalqueu}
        			\lim_{T\to \infty}\dfrac{1}{T}\sum_{t=0}^{T-1}\sum_{k\in \mathcal{K}}\mathbb{E}[\bar Q_k(t)]\le \dfrac{B+V\hat G(\epsilon)}{\epsilon},
        			\end{align}
        			where 
        		%	$G^{\text{opt}}$ denotes the optimal value of the time average total transmit power in the system, 
        			$ \hat G(\epsilon) $ with $\epsilon>0$ is specified by Assumption \ref{asum3} (i.e., \eqref{slaterassump1} and \eqref{slaterassump2}), 
        			%$ P^{\text{max}} $ indicates the maximum power that can be transmitted over one sub-channel in the system,
        			%\footnote{In this paper, we assume that the value of $P^{\text{max}}$ is large enough so that we do not have infallibility problem.}, 
        			and constant $ B $ is determined as follows
        			\begin{align}\label{equationb}
        			B=\dfrac{1}{2}\sum_{k\in\mathcal{K}}\bigg((\Delta^{\text{max}}_k)^2+(\delta^{\text{max}})^2\bigg).
        			\end{align}
        			%where $\delta^{\text{max}}$ is given in Theorem \ref{asum2}.
%        			\begin{align}
%        			\bar\epsilon=\min\left(1, \min_{\forall k\in \mathcal{K}}[\Delta_k^{\text{max}}-1.5]\right).
%        			\end{align}
        		\end{theorem}  
            	Before proving Theorem \ref{theo1}, we present the following remark.
        	\begin{rem}\normalfont       	
        		Inequality \eqref{optimalqueu} implies the strong stability of the virtual queues $\{Q_k(t)\}_{k\in \mathcal{K}}$ which, in turn, implies that  the average AoI constraints in \eqref{eqo5} are satisfied. In addition, from inequality \eqref{optimalqueu}, we can see that the upper bound of the average backlogs of the virtual queues is an increasing linear function of parameter $V$. Moreover,  
        		inequality \eqref{optimalobject} shows that  the value of $V$ can be chosen so that $\dfrac{B}{V}$ is arbitrarily small, and thus, the average total transmit power achieved by Algorithm~\ref{table-1} becomes arbitrarily close to the optimal value $G^{\text{opt}}$.  Consequently, parameter $V$ provides a trade-off between the optimality of the objective function (i.e., the average total transmit power) and average backlogs of the virtual queues in the system.     
        	\end{rem}
        	      	
        	      	Next, we prove Theorem \ref{theo1}.
           	\begin{proof}          		
%           Recall that in  \eqref{021mb40}, we presented an upper bound for the drift-plus penalty function \eqref{consta6} as follows
%           \begin{align}\nonumber
%           &\alpha(\mathcal{S}(t))+V\sum_{k\in \mathcal{K}}\sum_{n\in\mathcal{N}}\mathbb{E}[p_{k,n}(t)]\le V\sum_{k\in \mathcal{K}}\sum_{n\in\mathcal{N}}\mathbb{E}[p_{k,n}(t)]+\dfrac{1}{2}\sum_{k\in\mathcal{K}}\bigg((\Delta^{\text{max}}_k)^2\\&\label{021mb400}+\mathbb{E}[\delta^2_k(t+1)|\mathcal{S}(t)]+2Q_k(t)\big(\mathbb{E}[\delta_k(t+1)|\mathcal{S}(t)]-\Delta^{\text{max}}_k\big)\bigg).
%           \end{align}
          % Algorithm \ref{table-1} aims at solving problem \eqref{eqo1} by minimizing the right-hand side of \eqref{021mb40} over all possible policies.
            Let 
           %$\bar p_{k,n}(t)$ denote the allocated power to sensor $k$ over sub-channel $n$ determined by Algorithm \ref{table-1},
            $\bar \delta_k(t+1)$  and $\bar\alpha(\mathcal{S}(t)) $ denote the value of the AoI of sensor $k$ and the conditional Lyapunov drift as determined by Algorithm~\ref{table-1} in slot $t$, respectively. 
            %and  $\bar\alpha(\mathcal{S}(t)) $ denote the conditional Lyapunov drift determined by Algorithm \ref{table-1} .
             Then, by using the bound in \eqref{021mb40} and Lemma \ref{theo01}, we have
           %
          % 
           %the \textit{channel-only} policy that yields \eqref{cha00} and \eqref{cha01} for a fixed $\eta>0$, we have 
            \begin{equation}
           \begin{array}{ll}
           \bar\alpha(\mathcal{S}(t))+V\sum_{k\in \mathcal{K}}\sum_{n\in\mathcal{N}}\mathbb{E}[\bar p_{k,n}(t)\mid\mathcal{S}(t)]\overset{(a)}{\le} V\sum_{k\in \mathcal{K}}\sum_{n\in\mathcal{N}}\mathbb{E}[\bar p_{k,n}(t)\mid\mathcal{S}(t)]+\\\dfrac{1}{2}\sum_{k\in\mathcal{K}}\bigg((\Delta^{\text{max}}_k)^2+\mathbb{E}[(\bar\delta_k(t+1))^2\mid\mathcal{S}(t)]+2Q_k(t)\big(\mathbb{E}[\bar \delta_k(t+1)\mid\mathcal{S}(t)]-\Delta^{\text{max}}_k\big)\bigg)\overset{(b)}{\le}\\   V\mathbb{E}\sum_{k\in \mathcal{K}}\sum_{n\in\mathcal{N}}\left[p^*_{k,n}(t)\big|\mathcal{S}(t)\right]+\dfrac{1}{2}\sum_{k\in\mathcal{K}}\bigg((\Delta^{\text{max}}_k)^2+\mathbb{E}[(\delta^*_k(t+1))^2\mid\mathcal{S}(t)]+2\bar Q_k(t)\\\big(\mathbb{E}[\delta^*_k(t+1)\mid\mathcal{S}(t)]-\Delta^{\text{max}}_k\big)\bigg)\overset{(c)}{\le} V(G^{\text{opt}}+\nu)+\dfrac{1}{2}\sum_{k\in\mathcal{K}}\bigg((\Delta^{\text{max}}_k)^2+(\delta^{\text{max}})^2+2\bar Q_k(t)\nu\bigg),
  \end{array}
\label{021mb4001}
\end{equation}
           where, as defined earlier, $p^*_{k,n}(t)$ and $\delta_k^*(t+1)$  denote the allocated power to sensor $k$ over sub-channel $n$ and the value of the AoI of sensor $k$ determined by the  {channel-only} policy that yields \eqref{cha00} and \eqref{cha01} for a fixed $\nu>0$. Inequality  $ (a) $  comes from the upper bound in \eqref{021mb40}. Inequality $ (b) $ follows because i) Algorithm~\ref{table-1} minimizes the left-hand side of inequality  $ (b) $  over all possible policies (not only {channel-only} policies) by  using the opportunistically minimizing an expectation method \cite[Sect.~1.8]{neely2010stochastic} and ii) the considered {channel-only} policy that yields \eqref{cha00} and \eqref{cha01} is a particular policy  among all the  policies. Inequality  $ (c) $  follows because i) we have $ \mathbb{E}[(\delta^*_k(t+1))^2\mid\mathcal{S}(t)]\le (\delta^{\text{max}})^2$, and ii) the considered {channel-only} policy that yields \eqref{cha00} and \eqref{cha01} is independent of the network state $\mathcal{S}(t)$, thus we have 
           \begin{align}\label{cha001}
           &\sum_{k\in \mathcal{K}}\sum_{n\in\mathcal{N}}\mathbb{E}\left[p^*_{k,n}(t)\big|\mathcal{S}(t)\right]=\sum_{k\in \mathcal{K}}\sum_{n\in\mathcal{N}}\mathbb{E}\left[p^*_{k,n}(t)\right]\le G^{\text{opt}}+\nu,\\&\label{cha011}
           \mathbb{E}[\delta^*_k(t+1)\mid\mathcal{S}(t)]-\Delta^{\text{max}}_k=\mathbb{E}[\delta^*_k(t+1)]-\nu\le\Delta^{\text{max}}_k ,~ \forall k\in\mathcal{K}.
           \end{align}
           
           By taking $\nu\rightarrow 0$, \eqref{021mb4001} results in the following inequality
           \begin{align}\label{powerjadid}
           \bar\alpha(\mathcal{S}(t))+V\sum_{k\in \mathcal{K}}\sum_{n\in\mathcal{N}}\mathbb{E}[\bar p_{k,n}(t)\mid\mathcal{S}(t)]\le B+VG^{\text{opt}},
           \end{align}
where  $ B $ is the constant defined in \eqref{equationb}, i.e.,
$
B=1/2\sum_{k\in\mathcal{K}}\big((\Delta^{\text{max}}_k)^2+(\delta^{\text{max}})^2\big).
$

       Taking expectations over randomness of the network state at both sides of \eqref{powerjadid} and using the law of iterated expectations, we have 
       \begin{align}\label{lem1enequ120}
       \mathbb{E}\left[L\left(\bar {\bold{Q}}(t+1)\right)]-\mathbb{E}[L\left(\bar {\bold{Q}}(t)\right)\right]+V\sum_{k\in \mathcal{K}}\sum_{n\in\mathcal{N}}\mathbb{E}[\bar p_{k,n}(t)]\le B+VG^{\text{opt}},
       \end{align}     
       where $\bar{\bold{Q}}(t)$ denotes a vector containing all the virtual queues in slot $t$ under Algorithm~\ref{table-1}.   	
       By summing over $t\in\{0,\ldots,T-1\}$ and using the law of telescoping sums, we have 
       \begin{align}\label{prooflem110}
       &\mathbb{E}\left[L\left(\bar{\bold{Q}}(T)\right)\right]-\mathbb{E}\left[L\left(\bar{\bold{Q}}(0)\right)\right]+V\sum_{t=0}^{T-1}\sum_{k\in \mathcal{K}}\sum_{n\in\mathcal{N}}\mathbb{E}[\bar p_{k,n}(t)]\le TB+TVG^{\text{opt}}.
       \end{align} 
      
       Now we are ready to prove  the bound of the average total transmit power in \eqref{optimalobject}. In this regard, we rewrite \eqref{prooflem110} as follows
        \begin{equation}
       \begin{array}{ll}
       V\sum_{t=0}^{T-1}\sum_{k\in \mathcal{K}}\sum_{n\in\mathcal{N}}\mathbb{E}[\bar p_{k,n}(t)]\le-{\mathbb{E}\left[L\left(\bar{\bold{Q}}(T)\right)\right]}+\mathbb{E}\left[L\left(\bar{\bold{Q}}(0)\right)\right]+ TB+TVG^{\text{opt}}\overset{(a)}{\le}\\ TB+TVG^{\text{opt}}+\mathbb{E}\left[L\left(\bar{\bold{Q}}(0)\right)\right], 
  \end{array}
\label{prooflem120}
\end{equation}
       where inequality $ (a) $ follows because we neglected the negative term on the right-hand side of the first inequality. Dividing \eqref{prooflem120} by $TV$, we have 
       \begin{align}\label{prooflem130}
       &\dfrac{1}{T}\sum_{t=0}^{T-1}\sum_{k\in \mathcal{K}}\sum_{n\in\mathcal{N}}\mathbb{E}[\bar p_{k,n}(t)]\le\dfrac{B}{V}+G^{\text{opt}}+\dfrac{\mathbb{E}\left[L\left(\bar{\bold{Q}}(0)\right)\right]}{TV}. 
       \end{align} 
       Since $\mathbb{E}\left[L\left(\bar{\bold{Q}}(0)\right)\right]$ has a finite value, taking the limit   $T\rightarrow \infty$  in \eqref{prooflem130} proves the bound of the average total transmit power in \eqref{optimalobject}.
       
        To prove the bound of the average backlogs of the virtual queues in \eqref{optimalqueu}, we assume that the Slater's condition presented in Assumption \ref{asum3} holds. In other words,  we assume that there is a {channel-only} policy so that the virtual queues are strongly stable. Thus, by using the bound in \eqref{021mb40}, we have  (cf. \eqref{021mb4001})          
 %By using the {channel-only} policy that yields \eqref{slaterassump1} and \eqref{slaterassump2} in the slater assumption,  we have 
 \begin{equation}
\begin{array}{ll}
       \bar\alpha(\mathcal{S}(t))+V\sum_{k\in \mathcal{K}}\sum_{n\in\mathcal{N}}\mathbb{E}[\bar p_{k,n}(t)\mid\mathcal{S}(t)]\overset{(a)}{\le} V\sum_{k\in \mathcal{K}}\sum_{n\in\mathcal{N}}\mathbb{E}[\bar p_{k,n}(t)\mid\mathcal{S}(t)]+\\\dfrac{1}{2}\sum_{k\in\mathcal{K}}\bigg((\Delta^{\text{max}}_k)^2+\mathbb{E}[(\bar\delta_k(t+1))^2\mid\mathcal{S}(t)]+2\bar Q_k(t)\big(\mathbb{E}[\bar \delta_k(t+1)\mid\mathcal{S}(t)]-\Delta^{\text{max}}_k\big)\bigg)\overset{(b)}{\le} \\  V\sum_{k\in \mathcal{K}}\sum_{n\in\mathcal{N}}\mathbb{E}\left[\hat p_{k,n}(t)\big|\mathcal{S}(t)\right]+\dfrac{1}{2}\sum_{k\in\mathcal{K}}\bigg((\Delta^{\text{max}}_k)^2+\mathbb{E}[(\hat\delta_k(t+1))^2\mid\mathcal{S}(t)]+\\2\bar Q_k(t)\big(\mathbb{E}[\hat\delta_k(t+1)\mid\mathcal{S}(t)]-\Delta^{\text{max}}_k\big)\bigg)\overset{(c)}{\le} V\hat G(\epsilon)+B-\epsilon\sum_{k\in\mathcal{K}}\bar Q_k(t),
  \end{array}
\label{021mb40012}
\end{equation}
       where, as defined earlier, $\hat p_{k,n}(t)$ and $\hat\delta_k(t+1)$  denote the allocated power to sensor $k$ over sub-channel $n$ and the value of the AoI of sensor $k$ determined by the  {channel-only} policy that yields \eqref{slaterassump1} and \eqref{slaterassump2} in the Slater's condition, respectively. Inequality $ (a) $ comes from the upper bound in \eqref{021mb40}. Inequality $ (b) $ follows because i) Algorithm~\ref{table-1} minimizes the left-hand side of inequality $ (b) $ over all possible policies (not only channel-only policies) by  using the opportunistically minimizing an expectation method and ii) the considered {channel-only} policy that yields \eqref{slaterassump1} and \eqref{slaterassump2} in the Slater's condition is a particular policy  among all the policies. Inequality $ (c) $ follows because i) we have  $ \mathbb{E}[(\hat \delta_k(t+1))^2\mid\mathcal{S}(t)]\le (\delta^{\text{max}})^2$, ii) constant $B$ is given by \eqref{equationb}, and iii) the channel-only policy that yields \eqref{slaterassump1} and \eqref{slaterassump2} is independent of the network state $\mathcal{S}(t)$, thus we have 
       \begin{align}\label{asuumslat2}
       &\sum_{k\in \mathcal{K}}\sum_{n\in\mathcal{N}}\mathbb{E}\left[\hat p_{k,n}(t)\mid\mathcal{S}(t)\right]=\sum_{k\in \mathcal{K}}\sum_{n\in\mathcal{N}}\mathbb{E}\left[\hat p_{k,n}(t)\right]=\hat G(\epsilon),\\&\label{slaterassump2201}
       \mathbb{E}[\hat\delta_k(t+1)\mid\mathcal{S}(t)]+\epsilon=\mathbb{E}[\hat\delta_k(t+1)]+\epsilon\le\Delta^{\text{max}}_k,~ \forall k\in\mathcal{K}.
       \end{align}
       
   Taking expectations over randomness of the network state at both sides of the resulting inequality in \eqref{021mb40012} and using the law of iterated expectations, we have     
       	\begin{align}\label{lem1enequ1210}
       \mathbb{E}\left[L\left(\bar{\bold{Q}}(t+1)\right)]-\mathbb{E}[L\left(\bar{\bold{Q}}(t)\right)\right]+V\sum_{k\in \mathcal{K}}\sum_{n\in\mathcal{N}}\mathbb{E}[\bar p_{k,n}(t)]\le B-\epsilon\sum_{k\in \mathcal{K}}\mathbb{E}[\bar Q_k(t)]+V\hat G(\epsilon).
       \end{align} 
               	
       By summing over $t\in\{0,\ldots,T-1\}$ and using the law of telescoping sums, we have 
       \begin{equation}
      \begin{array}{ll}
       \mathbb{E}\left[L\left(\bar{\bold{Q}}(T)\right)\right]-\mathbb{E}\left[L\left(\bar{\bold{Q}}(0)\right)\right]+V\sum_{t=0}^{T-1}\sum_{k\in \mathcal{K}}\sum_{n\in\mathcal{N}}\mathbb{E}[\bar p_{k,n}(t)]\le\\ TB-\epsilon\sum_{t=0}^{T-1}\sum_{k\in \mathcal{K}}\mathbb{E}[\bar Q_k(t)]+TV\hat G(\epsilon).
     \end{array}
\label{prooflem1110}
   \end{equation}
       
       To prove the bound 
       %of the time average backlog of the virtual queues
       in \eqref{optimalqueu}, we rewrite \eqref{prooflem1110} as follows
\begin{equation}
\begin{array}{ll}
       \epsilon\sum_{t=0}^{T-1}\sum_{k\in \mathcal{K}}\mathbb{E}[\bar Q_k(t)]\le-\mathbb{E}\left[L\left(\bar{\bold{Q}}(T)\right)\right]+\mathbb{E}\left[L\left(\bar{\bold{Q}}(0)\right)\right]-V\sum_{t=0}^{T-1}\sum_{k\in \mathcal{K}}\sum_{n\in\mathcal{N}}\mathbb{E}[\bar p_{k,n}(t)]\\+TB+TV\hat G(\epsilon)\overset{(a)}{\le} TB+TV\hat G(\epsilon)+\mathbb{E}\left[L\left(\bar{\bold{Q}}(0)\right)\right],
           \end{array}
       \label{prooflem14}
       \end{equation} 
       where inequality $ (a) $ follows because  we neglected the negative terms on the right-hand side of the first inequality. 
       %and ii) the maximum value of total transmit power in each slot is $\sum_{k\in\mathcal{K}}P_k^{\text{max}}$, and $G^*\le\sum_{k\in\mathcal{K}}P_k^{\text{max}}$.
       By dividing \eqref{prooflem14} by $T\epsilon$, we have 
       \begin{align}\label{prooflem15}
       &\dfrac{1}{T}\sum_{t=0}^{T-1}\sum_{k\in \mathcal{K}}\mathbb{E}[\bar Q_k(t)]\le \dfrac{B+V\hat G(\epsilon)}{\epsilon}+\dfrac{\mathbb{E}\left[L\left(\bar{\bold{Q}}(0)\right)\right]}{T\epsilon }.
       \end{align} 
       Since $\mathbb{E}\left[L\left(\bar{\bold{Q}}(0)\right)\right]$ has a finite value, taking the limit   $T\rightarrow \infty$  in \eqref{prooflem15} proves the bound of the average backlogs of the virtual queues in \eqref{optimalqueu}.
                  \end{proof}

\section{A Sub-optimal Solution for the Per-Slot  Problem \eqref{eqoi1} }\label{A Suboptimal Solution to Solve Problem}
    As discussed in Section \ref{Solution Algorithm}, we need to solve an instance of the optimization problem \eqref{eqoi1} in each slot (see Algorithm~\ref{table-1}). Problem \eqref{eqoi1} is a mixed integer non-convex optimization problem containing both integer (i.e., sub-channel assignment and sampling action) and continuous (i.e., power allocation) variables. Thus, finding its optimal
    solution is not trivial and conventional methods for solving convex optimization problems cannot directly be used. 
    The optimal solution of problem \eqref{eqoi1} can be found by an exhaustive search method that requires searching over all possible combinations of the binary variables, i.e., the sub-channel assignment and sampling action variables, and solving a (simple) power allocation problem for each such combination. However, the computational complexity of this method increases exponentially with the number of sampling action and sub-channel assignment variables in the system (i.e., $ KN $). Thus, finding an appropriate  sub-optimal solution with low computational complexity is necessary for the optimization problem \eqref{eqoi1}.       
    %
%    
%    
%    Thus, for each possible sampling actions of  sensors  a low-complexity sub-optimal solution for the joint power allocation and sub-channel assignment is proposed. 
%
    Next,  in  Section \ref{Solution Algorithm02}, we present the proposed sub-optimal solution. 
    Then, in Section \ref{Complexity of the Proposed Sub-Optimal Solution}, we present  complexity analysis of the sub-optimal solution.
    \subsection{Solution Algorithm}\label{Solution Algorithm02}

        The main idea behind our proposed sub-optimal solution is to reduce the computational complexity from that of the full exhaustive search method described above. To this end,  we search only over all possible combinations of sampling action variables $ b_k(t), \forall k\in\mathcal{K} $; for each such combination,
        %Thus, the complexity grows exponentially in the number of sampling action variables $ K $  as compared to being exponential in $ KN $ as for the exhaustive search method. 
    %This reduction can be significant in practical scenarios. 
    % 
    % 
    %    of the exhaustive search method by reducing the number of variables that we need to search over them. With the proposed sub-optimal solution, the number of variables that we need to search over them  equals to the number of sampling action variables. 
%  for each of the $ 2^K $ possible combinations of sampling actions for which number of sensors that have a sample to transmit  is less than or equal to the  number of sub-channels, 
%  we are faced to solve a joint power allocation and sub-channel assignment problem, which is still a mixed integer non-convex optimization problem. Next, 
  we propose a low-complexity two-stage optimization strategy to find a sub-optimal solution to the joint power allocation and sub-channel assignment problem. Then, among all the solutions, the best one is selected as the sub-optimal solution to  problem \eqref{eqoi1}. Note that if the number of sub-channels $N$ is less than the number of sensors $K$ (i.e., $N<K$) we do not search over all possible  combinations of sampling action variables because the maximum number of sensors that can take a sample in each slot is $N$. 
  %Thus, the combinations for which number of sensors that have a sample to transmit  is less than or equal to the  number of sub-channels are the combinations that we need to search over them.  
    
%   The main idea behind our proposed sub-optimal solution is to reduce the number of parameters that we need to search over them. In the proposed sub-optimal solution, we search over all possible sampling actions of the sensors and for a given sampling action policy in the system, we propose a sub-optimal solution for the joint power allocation and sub-channel assignment.
 
     Let  ${\bold{b}(t)=[b_{1}(t),\ldots,b_{K}(t)]}$ denote a vector containing all binary sampling action variables in slot $t$. Further, let ${\mathcal{B}}$ denote the set of all possible values of binary vector  $\bold{b}(t)$ with cardinality $|\mathcal{B}|=2^K$.  In addition, let $\tilde{\mathcal{B}}\subseteq \mathcal{B}$ denote the set of  all possible values of such binary vectors  $\bold{b}(t)$ for which the number of sensors that have a sample to transmit is less than or equal to the  number of sub-channels $N$, i.e., $\tilde{\mathcal{B}}=\{\bold{b}(t)\mid\bold{b}(t)\in\mathcal{B},\|\bold{b}(t)\|_0\le N\},$ where $\|\cdot\|_0$ counts the number of non-zero elements in a vector. Note that for  $K\le N,$ the set $\tilde{\mathcal{B}}$ is equal to set ${\mathcal{B}}$, i.e., $\tilde{\mathcal{B}}=\mathcal{B}$.
     
     The steps of the proposed sub-optimal solution are summarized in Algorithm~\ref{table-2_1}.  
     Step~2 performs exhaustive search over feasible  sampling actions $\bold{b}(t)\in \tilde{\mathcal{B}}$:
%     those vectors $\bold{b}(t)$ for which the number of  sensors that have a sample to transmit  is less than or equal to the  number of sub-channels: 
     for each  such $\bold{b}(t)$ a sub-optimal power allocation and sub-channel assignment is obtained by finding an approximate solution for the mixed integer non-convex problem \eqref {eqoi1} with variables $\{\rho_{k,n}(t),p_{k,n}(t)\}_{k \in\mathcal{K}, n \in\mathcal{N}}$ (i.e., problem \eqref{power1c010}) which is presented in the next subsections. Step~3 returns a sub-optimal solution to problem  \eqref{eqoi1}.

    \begin{algorithm}
    	{   \caption{Proposed sub-optimal solution algorithm to problem  \eqref{eqoi1}  }
    		\label{table-2_1}
%    		For $ i=1: |\mathcal{B}| $
%    		1:  \textbf{Power allocation and sub-channel assignment}:  
             Step 1.  \textbf{Initialization}: set ${O=0}$, ${\{\tilde{b}_k(t)= 0\}_{k \in\mathcal{K}}}$, ${\{\tilde{p}_{k,n}(t)=0\}_{k \in\mathcal{K}, n \in\mathcal{N}}}$, and ${\{\tilde{\rho}_{k,n}(t)=0\}_{k \in\mathcal{K}, n \in\mathcal{N}}}$ \\
    		 Step 2.  For each $\bold{b}(t)\in \tilde{\mathcal{B}}$  \textbf{do}	            

    		   	 $~~~~~~~~~$A. Find a sub-optimal solution for the following joint power allocation and sub-channel 
    		   	 $~~~~~~~~~~~~~~$assignment problem 
    		   	     	  \begin{equation}
    		   	     	\begin{array}{ll}
    		   	     	\mbox{minimize}   &      V\sum_{k\in \mathcal{K}}\sum_{n\in\mathcal{N}}p_{k,n}(t)+\dfrac{1}{2}\sum_{k\in\mathcal{K}} 
    		   	     	b_k(t)\left[1-(\delta_k(t)+1)^2-2Q_k(t)\delta_k(t)\right]\\ 
    		   	     \mbox{subject to}& 
    		   	     	\eqref{eq1o5}-\eqref{eq1o50},
    		   	      \end{array}
                      \label{power1c010}
    		   	    \end{equation}
    		   	 $~~~~~~~~~~$with variables $\{\rho_{k,n}(t),p_{k,n}(t)\}_{k \in\mathcal{K}, n \in\mathcal{N}}$\\
    		   	 $~~~~~~~~~~$B. Denote the obtained solution by $\{\dot{p}_{k,n}(t)\}_{k \in\mathcal{K}, n \in\mathcal{N}}$, $\{\dot{\rho}_{k,n}(t)\}_{k \in\mathcal{K}, n \in\mathcal{N}}$, and $\{\dot{ b}_{k}(t)\}_{k \in\mathcal{K}}$  \\
    		   	  %$~~~~~~~~$\textbf{A:} Find a sub-optimal solution for the following optimization problem 
%    		    $~~~~~~~~~~~~$ 	I. Initialize sets $  \mathcal{K}'=\{k | k\in\mathcal{K}, b_{k}(t)=1  \}  $, $\mathcal{N}'=\mathcal{N}$, ${\{\rho_{k,n}(t)=0\}_{k \in\mathcal{K}, n \in\mathcal{N}}},$  and set  	$~~~~~~~~~~~~~~$$ i=1 $\\
%    		   	$~~~~~~~~~~~~$\textbf{ While} {$ i\le N $} \textbf{do}\\
%    		   	$~~~~~~~~~~~~~~$II. Set $ \rho_{k,n}=1$ where $(k,n)=\argmax_{k\in\mathcal{K}',n\in\mathcal{N}'}|h_{k,n}(t)|^2 $\\
%    		   	$~~~~~~~~~~~~~~$III. $\mathcal{N}'= \mathcal{N}'-\{n\} $\\
%    		   	$~~~~~~~~~~~~~~$IV. If $ \mathcal{K}'-\{k\} $ is empty set $\mathcal{K}'= \mathcal{K} $, else set $ \mathcal{K}'=\mathcal{K}'-\{k\} $\\
%    		   	%$~~~~~~~~~~~~~~$V.  If $ \mathcal{K}'$ is empty, set $\mathcal{K}'= \mathcal{K} $\\
%%    		   	$~~~~~~~~~~~~~~~~~~~~$  {\textbf{Else }  }\\
%%    		   	$~~~~~~~~~~~~~~~~~~$  {Go back to Step 1}\\
%    		   	$~~~~~~~~~~~~~~$V. $i= i+1$\\
%%    		   	%	$~~~~~~~~~~~~~~~$	\textbf{ End if}\\
%%    		   	$~~~~$\textbf{ End while}\\
%    		 	$~~~~~~~~$\textbf{B: Power allocation}
   	  $~~~~~~~~~~$C. If $ V\sum_{k\in \mathcal{K}}\sum_{n\in\mathcal{N}}\dot{p}_{k,n}(t)+\dfrac{1}{2}\sum_{k\in\mathcal{K}} 
   	  \dot{b}_k(t)\left[1-(\delta_k(t)+1)^2-2Q_k(t)\delta_k(t)\right]\le O$:\\
   	  $~~~~~~~~~~~~~$ I. Set $\{ \tilde{p}_{k,n}(t)=\dot{p}_{k,n}(t)\}_{k \in\mathcal{K}, n \in\mathcal{N}}$, $\{ \tilde{\rho}_{k,n}(t)=\dot{\rho}_{k,n}(t)\}_{k \in\mathcal{K}, n \in\mathcal{N}}$, $\{\tilde{b}_k(t)=\dot{ b}_{k}(t)\}_{k \in\mathcal{K}}$ \\
   	 $~~~~~~~~~~~~~$ II. Set $O= V\sum_{k\in \mathcal{K}}\sum_{n\in\mathcal{N}}\tilde{p}_{k,n}(t)+\dfrac{1}{2}\sum_{k\in\mathcal{K}} 
   	 \tilde{b}_k(t)\left[1-(\delta_k(t)+1)^2-2Q_k(t)\delta_k(t)\right] $\\
   	 Step 3.   Return  	 $\{ \tilde{\rho}_{k,n}(t)\}_{k \in\mathcal{K}, n \in\mathcal{N}}$, $\{ \tilde{\rho}_{k,n}(t)\}_{k \in\mathcal{K}, n \in\mathcal{N}}$, and  $\{\tilde{b}_k(t)\}_{k \in\mathcal{K}}$ as a sub-optimal solution to  problem  \eqref{eqoi1}
	
    	}
    \end{algorithm}

  The power allocation and sub-channel assignment problem  \eqref{power1c010} is a mixed integer non-convex optimization problem. Thus, 
  we propose a two-stage sequential optimization method to find a sub-optimal solution to \eqref{power1c010}.  The method performs first a greedy sub-channel assignment, which is followed by power allocation. 
     %First,  a greedy algorithm is proposed to assign the  sub-channels which is presented in Section \ref{Sub-channel Assignment}. Then, for the given sub-channel assignment we use the water-filling approach to determine the power allocation, as presented in Section \ref{Power Allocation}. 
    
    \subsubsection{ Sub-channel Assignment}\label{Sub-channel Assignment}
     The proposed greedy algorithm to assign the  sub-channels is presented in Algorithm \ref{table-2_2}.
     The main idea is to find the strongest sub-channel among all the sensors that have a sample to transmit, and assign this sub-channel greedily to that sensor (Step~1). This assigned sub-channel is then removed from the set of available  sub-channels  because each sub-channel can be assigned to at most one sensor (Step~2). For fairness, the sensor that was just assigned the sub-channel is removed from the set of competing sensors guaranteeing that this sensor cannot get more sub-channels until the other sensors get the same number of sub-channels (Step~3). This procedure is repeated until all the sub-channels are assigned to the sensors. 
     
%     
%     
%     The main idea behind the proposed greedy algorithm to assign the sub-channels is to find a sensor with the best channel situation among all the sensors and assign one sub-channel to it, and when a sensor gets one sub-channel it cannot get more until the other sensors get the same number of sub-channels.
%    In Step 1, a sensor having a sample to transmit  with best channel situation is determined and one sub-channel is assign to it.  Step 2 ensures that a sub-channel is assigned to at most one sensor in each slot. Steps 3 ensures that after assigning a sub-channel to a sensor  it cannot get more sub-channels until the other sensors get the same number of sub-channels in each slot. 
    %and IV ensure that when number of sub-channels is grater than the number of sensors (i.e., $N>K$) first one sub-channel is assigned to each sensor, then the rest of the  each sensor
     %and III, the assigned sub-channel and the user ar taken out from the 
    \begin{algorithm}
    	   \caption{Sub-Channel Assignment}
    		\label{table-2_2}
    	 \textbf{ Initialization:} a) initialize sets $  \mathcal{K}'=\{k \mid k\in\mathcal{K}, b_{k}(t)=1\},$ and  $\mathcal{N}'=\mathcal{N}$, b) initialize $\{\rho_{k,n}(t)=0\}_{k \in\mathcal{K}, n \in\mathcal{N}},$ and c) set $ i=1 $\\
    		$~~~~$\textbf{ While} {$ i\le N $} \textbf{do}\\
    		$~~~~~~~~~~~~$Step 1. Set $ \rho_{k,n}(t)=1$ where $(k,n)=\argmax_{k\in\mathcal{K}',n\in\mathcal{N}'}|h_{k,n}(t)|^2 $\\
    		$~~~~~~~~~~~~$Step 2. $\mathcal{N}'= \mathcal{N}'\setminus\{n\} $\\
    		$~~~~~~~~~~~~$Step 3. If $ \mathcal{K}'\setminus\{k\}=\emptyset,$  set $\mathcal{K}'= \mathcal{K} $; otherwise, set $ \mathcal{K}'=\mathcal{K}'\setminus\{k\} $\\
  	    %	$~~~~~~~~~~~~$Step 4.  {\textbf{If }$ \mathcal{K}'$ is empty } $\mathcal{K}'= \mathcal{K} $\\
%    	    $~~~~~~~~~~~~~~~~~~~~$  {\textbf{Else }  }\\
%    	    $~~~~~~~~~~~~~~~~~~$  {Go back to Step 1}\\
    	   	$~~~~~~~~~~~~$Step 4. $i= i+1$\\
    	%	$~~~~~~~~~~~~~~~$	\textbf{ End if}\\
    	    $~~~~$\textbf{ End while}

    \end{algorithm}

%    The optimization problem \eqref{power10} is a nonconvex problem because the equality constraint \eqref{eq1o52n} is not affine.  To tackle this issue, we change  constraint \eqref{eq1o52n} by the following constraint
%    \begin{align}
%    \sum_{n\in\mathcal{N}}\rho_{k,n}(t)r_{k,n}(t)\ge \eta b_{k}(t),  k\in \mathcal{K}.
%    \end{align}
%    Since the objective function of the optimization problem  \eqref{power10}  is an increasing function of $p_{k,n}(t)$,  it can be shown (e.g., by contradiction) constraint  \eqref{eq1o52n} holds by equality at the optimal point. Therefore, we solve the following optimization instead of optimization problem \eqref{power10}
%    \begin{subequations}\label{powevr10}
%    	\begin{align}
%    	&\text{minimize} \,\,\sum_{k\in \mathcal{K}}\sum_{n\in\mathcal{N}}p_{k,n}(t)\\&\label{eqvio21} 
%    	\text{subject\,\,to}
%%    	\hspace{.15cm}
%%    	\sum_{n\in\mathcal{N}}p_{k,n}(t)\le P_k^{\text{max}},\,\,\, k\in \mathcal{K}
%%    	\\&
%    	\hspace{.15cm} \sum_{n\in\mathcal{N}}\rho_{k,n}(t)r_{k,n}(t)\ge \eta b_{k}(t),  k\in \mathcal{K},
%    	\end{align}
%    \end{subequations}
%    with variables $\{p_{k,n}(t)\}_{k \in\mathcal{K}, n \in\mathcal{N}}$.
%    The optimization problem \eqref{powevr10} is a convex problem. 
    %To solve \eqref{powevr10} we use available online software such as CVX \cite{cvxonline}.
      \subsubsection{Power Allocation}\label{Power Allocation}
%      After assigning the sub-channels according to Algorithm \ref{table-2_2}, the remaining task is to allocate power to the assigned sub-channels. The power allocation problem is formulated as follows   
%          		\begin{subequations}\label{power10}
%          			\begin{align}
%          			&\text{minimize}~~\sum_{k\in \mathcal{K}}\sum_{n\in\mathcal{N}}p_{k,n}(t)
%          			%+\sum_{k\in\mathcal{K}} b_k(t)\bigg(1-(\delta_k(t)+1)^2-(\delta_k(t)+1)+2Q_k(t)\bigg)
%          			\\&\label{eq1o52n}
%          			\text{subject\,\,to}
%          			%\hspace{.15cm}
%          			%	\sum_{n\in\mathcal{N}}p_{k,n}(t)\le P_k^{\text{max}},\,\,\, k\in \mathcal{K}
%          			%\\&\label{eq1o52n}
%          			\hspace{.15cm} \sum_{n\in\mathcal{N}}\rho_{k,n}(t)W\log_2\left(1+\gamma_{k,n}(t)\right)= \eta b_{k}(t), \forall k\in \mathcal{K}
%          			\\&\label{eq1o500012}
%          			\hspace{1.85cm}p_{k,n}(t)\ge 0,~\forall  k\in \mathcal{K}, n\in \mathcal{N},
%          			\end{align}
%          		\end{subequations}
%          		with variables $\{p_{k,n}(t)\}_{k \in\mathcal{K}, n \in\mathcal{N}}$.
          		
          		Given that the  sub-channels have been assigned, power allocation for each sensor that has a sample to transmit can be determined separately. Let $\mathcal{N}_k\subseteq\mathcal{N}$ denote the set of sub-channels assigned to sensor $k$. 
          		%According to Algorithm \ref{table-2_2}, when $b_k(t)=0,$ we have $\mathcal{N}_k=\emptyset$. 
          		Thus, for each sensor $k$ that has a sample to transmit (i.e., $ b_k(t)=1 $), the following optimization problem needs to be solved
 \begin{equation}
\begin{array}{ll}
\mbox{minimize}   &   \sum_{n\in\mathcal{N}_k}p_{k,n}(t)
      	\\
      	\mbox{subject to}& \sum_{n\in\mathcal{N}_k}\rho_{k,n}(t)W\log_2\left(1+\dfrac{p_{k,n}(t)|h_{k,n}(t)|^2}{WN_0}\right)= \eta
      	\\
      	&p_{k,n}(t)\ge 0,~\forall  n\in \mathcal{N}_k,
   \end{array}
\label{power101}
\end{equation}
  with variables $\{p_{k,n}(t)\}_{ n \in\mathcal{N}_k}$. The optimization problem \eqref{power101} can be solved by  the water-filling approach \cite[Proposition~2.1]{7341067}.
  % We use the water-filling approach to solve the power allocation in the numerical results section.
  % Let $d_{k,n'}=\dfrac{WN_0}{|h_{k,n'}(t)|^2}$
%    The optimization problem \eqref{power1c0} is a convex problem. 
%    %an integer linear program that can be solved by the available online software such as MOSEK \cite{cvxonline}. 
%    

    \subsection{Complexity of the Proposed Sub-Optimal Solution }\label{Complexity of the Proposed Sub-Optimal Solution}
    In this section, we investigate the complexity of the proposed sub-optimal solution for problem~\eqref{eqoi1} and compare it with that of the full exhaustive search method.
   The proposed sub-optimal solution presented in Algorithm \ref{table-2_1} has three main steps, namely, i) determining the sampling actions which is solved by searching over all feasible sampling action combinations, ii) sub-channel assignment which is solved by the proposed greedy algorithm presented in Algorithm \ref{table-2_2}, and iii) power allocation which is solved by the water-filling approach. 
   %Since there are $K$ binary  sampling action variables, 
   The computational complexity  of the search over feasible sampling actions is equal to the cardinality of $\tilde{\mathcal{B}}$, i.e.,  $|\tilde{\mathcal{B}}|$ which is less than or equal to $2^K$ (recall that since $\tilde{\mathcal{B}}\subseteq \mathcal{B}$, we have  $|\tilde{\mathcal{B}}|\le|\mathcal{B}|=2^K$). Since the proposed greedy algorithm to solve the sub-channel assignment (i.e., Algorithm \ref{table-2_2}) has $N$ iterations, its   computational complexity is $N$. Since the water-filling approach needs at most $N$ iterations, its worst-case computational complexity is $N$ \cite{7341067}. Thus, since for each possible sampling action combination,  a sub-channel assignment and a  power allocation problem with complexity $2N$ is solved, 
   the  computational complexity  of the proposed sub-optimal solution presented in Algorithm~\ref{table-2_1} is $\mathcal{O}=2N|\tilde{\mathcal{B}}|$. Next, we investigate the computational complexity of the exhaustive search method to solve problem~\eqref{eqoi1}.
   
The exhaustive search method searches over all feasible binary sampling action and sub-channel assignment variables. For each such combination, a convex power allocation problem is solved. 
Since the computational complexity  of the search over sampling actions is $|\tilde{\mathcal{B}}|$
and there are $NK$ binary  sub-channel assignment variables, computational complexity of the binary search  is $\mathcal{O}=|\tilde{\mathcal{B}}|2^{KN}$. Assuming that the water-filling approach presented in \cite[Proposition~2.1]{7341067} is used to solve the power allocation problem, %(similarly as in Algorithm \ref{table-2_1}), 
the worst-case computational complexity of the power allocation is $N$. Thus, the computational complexity of the exhaustive search method  to solve problem \eqref{eqoi1} is  $\mathcal{O}=N|\tilde{\mathcal{B}}|2^{KN}$. 

%By comparing the computational complexity of the proposed sub-optimal solution to that of the exhaustive search method
Considering the discussion above, we can see that as compared to the full exhaustive search method, the computational complexity of the proposed sub-optimal solution reduces by a factor that is exponential in $KN$.

        \section{Numerical and Simulation Results}\label{Numerical Results}
        In this section, we evaluate the performance of the proposed dynamic control algorithm presented in Algorithm~\ref{table-1} in terms of transmit power consumption and AoI of the sensors. In addition, we evaluate the optimality gap  of the sub-optimal solution applied to solve problem~\eqref{eqoi1} presented in  Algorithm~\ref{table-2_1}. 
        \subsection{Simulation Setup}        
        We consider a WSN depicted in Fig. \ref{Location}, where the sink is located in the center and ${K=10}$ sensors are randomly placed  in a two-dimensional plane.
        %Due to the complexity of the exhaustive search solution used to solve the optimization problem \eqref{eqoi1}, we evaluate the performance of the system with a small number of sensors and sub-channels. 
        %We consider one sink in the center and ${K=10}$ sensors randomly placed  in a two-dimensional plane as depicted in Fig. \ref{Location}. 
        Sensors are indexed  according to their distance to the sink in such a way that sensor $ 1 $ is the nearest sensor to the sink and sensor $ 10 $ is the farthest. 
        %and ${N=2}$ sub-channels with
        %The  coordinate of sensor $1$ is $(0,300)$, the coordinate of sensor $2$ is  $(300,0)$, and the coordinate of the sink is $(0,0)$.  
        The channel coefficient from sensor $k$ to the sink over sub-channel $n$ in slot $t$ is modeled as ${h_{k,n}(t)=(d_k/d_0)^\xi c_{k,n}(t)}$, where $d_k$ is the distance from sensor $k$  to the sink, $ d_0 $ is the
        far field reference distance,  $\xi$ is the path loss exponent,
        and $c_{k,n}(t)$ is a Rayleigh distributed random coefficient.   Accordingly,  $(d_k/d_0)^\xi$  represents  large-scale fading and the term ${c_{k,n}(t)}$ represents small-scale Rayleigh fading. We set ${\xi=-3}$, $ {d_0=1} $, and the parameter of Rayleigh distribution as $0.5$.  The bandwidth of each sub-channel is  ${W=180}$ kHz. The size of each  packet is ${\eta=600}$ Bytes. We set ${\Delta_{k}^{\text{max}}=\Delta^{\text{max}}, \forall k}$. 
        %We use $\Delta^{\text{max}}=4$ and $ N=10 $ sub-channels, unless otherwise stated.
         \begin{figure}
        	\centering
        	\includegraphics[scale=0.35]{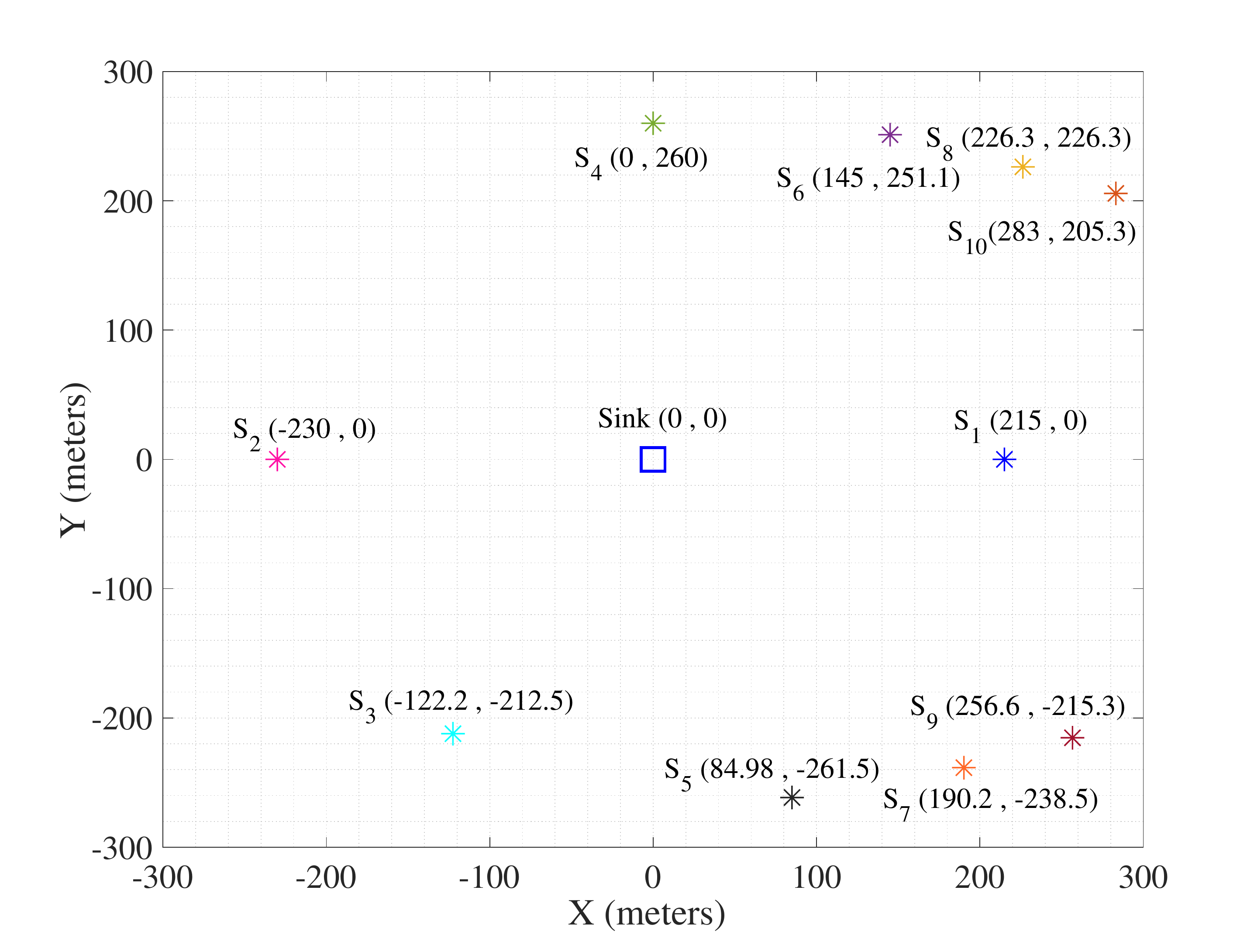}
        	\caption{The considered WSN where the sink is located in the center and ${K=10}$ sensors are randomly placed. The coordinates of sensor $k$ is shown by $S_k(x_k,y_k)$. }
        	\label{Location}
        		\vspace{-10mm}
        \end{figure}
        %Let  $T_{\text{max}}$ denote  the number of slots that we run  Algorithm \ref{table-1}.
          \subsection{Performance of the Proposed Dynamic Control Algorithm}
          In this section, we evaluate the performance of the proposed dynamic control algorithm (i.e., Algorithm~\ref{table-1}) in terms of transmit power consumption and AoI of the sensors. To solve the optimization problem \eqref{eqoi1}, we use Algorithm \ref{table-2_1}.
          
          Fig.~\ref{PowerN10K10D4} illustrates the evolution of the  average total transmit power for different values of parameter  $V$ with maximum acceptable average AoI of sensors ${\Delta^{\text{max}}=4}$ and $N=10$ sub-channels. The figure shows that when  $V$ increases, the average total transmit power decreases. This is because  when $V$ increases, more emphasis is set to minimize the total transmit power in the objective function of optimization problem \eqref{eqoi1}. 
          
          In addition, for the benchmarking, we consider a baseline policy that has a fixed sampling rate. The sampling rate is set as $ 1/7 $, so that resulting average AoI of each sensor is equal to the maximum acceptable average AoI ${\Delta^{\text{max}}_k=\Delta^{\text{max}}=4, \forall k\in \mathcal{K}}$. The sampling schedule of the considered baseline method is presented in Table~\ref{table:schedule}. For this baseline policy,  the sub-channel assignment and transmit power allocation are determined by the  proposed methods in Sections \ref{Sub-channel Assignment} and  \ref{Power Allocation}, respectively.  
          As it can be seen in Fig.~\ref{PowerN10K10D4}, the proposed dynamic control algorithm achieves more than $60~\%$  saving in the average total transmit power compared to the baseline policy.
%          
%          As it can be seen,  by performing the proposed dynamic control algorithm more than $60~\%$  saving in the average total transmit power can be achieved.
          %
         % significantly outperforms the baseline method for a large value of $V$ more than $60\%$. 
          %In addition, we can see that there are some jumps in the  average total transmit power of the baseline method. 
%          This is because  we enforce a sensor to take a sample in  pre-specified slots according to a fixed schedule even if the sensor has a bad channel situation in those slots.         
This shows the advantage of the proposed dynamic control algorithm in optimizing the sampling process in contrast to relying on a pre-defined sampling schedule which enforces a sensor to transmit a status update even under a bad channel situation.

\begin{table}[t!]
\caption{Sampling schedule of the considered fixed rate sampling policy of rate $1/7$}\centering
\begin{tabular}{ |c|c|c|c|c|c|c|c|c|c|c|c|c|c|c|c|c| } 
\hline
Time slot $t$ & 1 & 2 & 3 & 4 & 5 & 6 & 7 & 8 & 9 & 10 & 11 & 12 & 13 & 14 & 15 & 16\\ 
\hline
Sensor $k$ with $b_k(t)=1$ & 1 & 2 & 3 & 4 & 5, 6 & 7, 8 & 9, 10 & 1 & 2 & 3 & 4 & 5, 6 & 7, 8 & 9, 10 & 1 & $\cdots$ \\  
\hline
\end{tabular}
\label{table:schedule}
\end{table}
          
%          for the benchmarking, we compare the time average total transmit power of the proposed dynamic algorithm by that of a fixed sampling rate policy in which  sampling rate is  $ 1/7 $ for all the sensors (i.e., each sensor takes a sample after any $ 7 $ slots). Note that we consider the sampling rate is  $ 1/7 $ because under this sampling rate the average AoI of each sensor would become ${\Delta_k=4, \forall k\in \mathcal{K}},$ which is equal to the considered  maximum acceptable time average AoI of sensors ${\Delta^{\text{max}}=4}$.
          % In addition, we can see that,  for the high values of $V$,  the average total transmit power of the sensors starts to saturate into a certain level.            
        \begin{figure}
        	\centering
        	\includegraphics[scale=0.47]{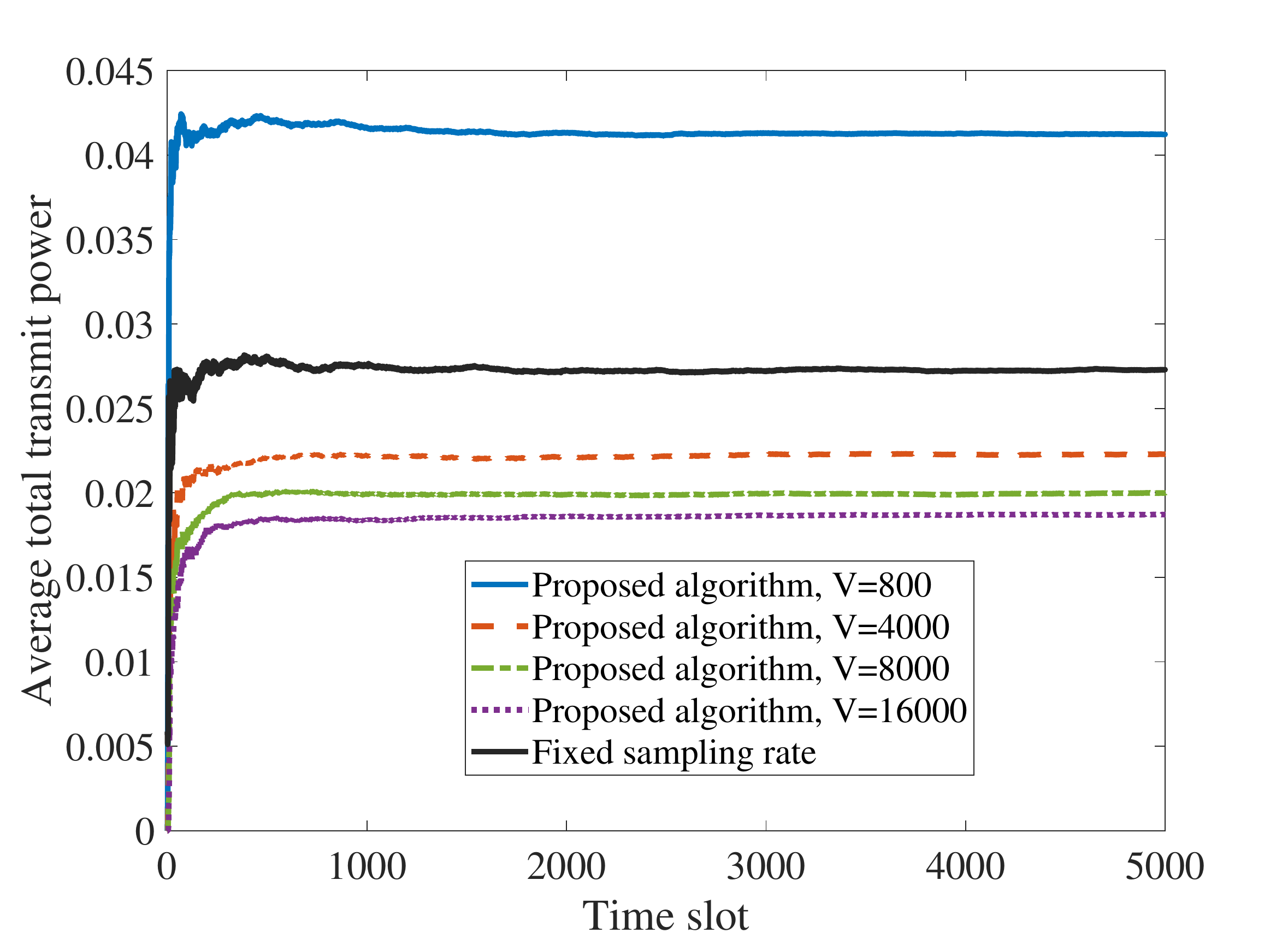}
        	\caption{ Evolution of the average total transmit power of the sensors for different values of $V$ with ${\Delta^{\text{max}}=4}$ and $N=10$. For comparison, a control policy with a fixed sampling rate is included as a baseline method.}
        	\vspace{-9mm}
        	\label{PowerN10K10D4}
        \end{figure}

Fig. \ref{lTradeoff} illustrates the trade-off between the average total transmit power and average  backlogs of the virtual queues as a function of $V$ for  ${\Delta^{\text{max}}=4}$ and $N=10$ sub-channels. As it can be seen, by increasing  $V$ the  average backlogs of the virtual queues increase and the average total transmit power decreases. This shows the inherent trade-off provided by the drift-plus-penalty method which was shown in Theorem \ref{theo1}. Moreover, the figure demonstrates that when $ V $ is sufficiently large, increasing it further does not significantly reduce the power. This is visible in Fig. \ref{PowerN10K10D4} as well.
\begin{figure}
	\centering
	\includegraphics[scale=0.44]{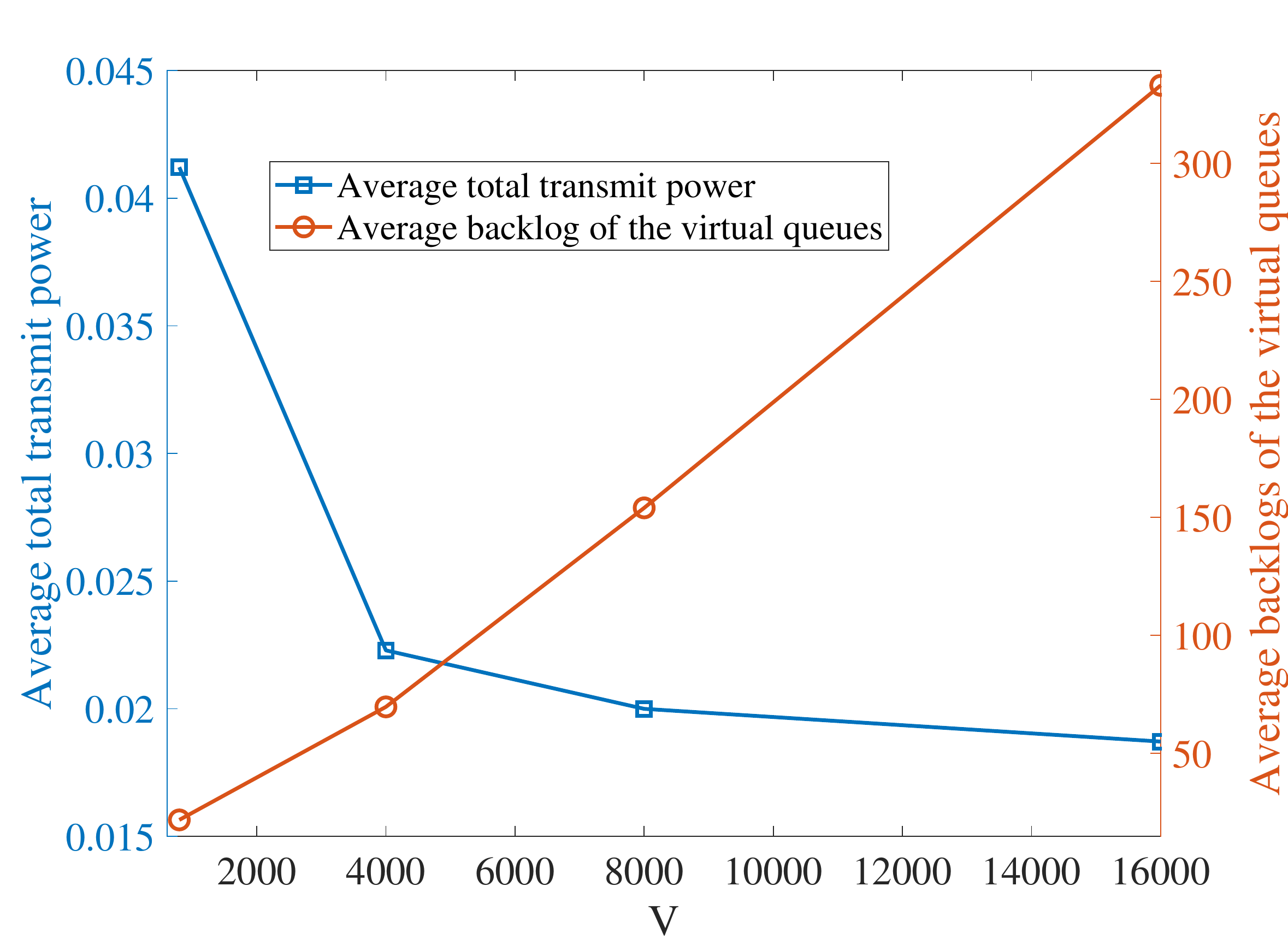}
	\caption{Trade-off between the average total transmit power of the sensors and average  backlogs of the virtual queues as a function of $V$. }
	\vspace{-9mm}
	\label{lTradeoff}
\end{figure}

 Fig. \ref{lPower101520v800} illustrates the evolution of the average total transmit power for different  numbers of sub-channels $N$ with  ${\Delta^{\text{max}}=4}$ and $V=8000$. The figure shows that when  $N$ increases, the average total transmit power decreases, as expected. This is because  when $N$ increases, more sub-channels can be assigned to each sensor, and thus, one packet can be transmitted with less power. 
In addition, we can see that the effect of increasing the number of sub-channels from $N=6$ to $N=8$, and further to $N=10$, is more profound. This is because when there are fewer sub-channels than sensors, due to the orthogonality of  the sub-channel assignment, all the sensors cannot be served in every slot and some sensors can get a sub-channel only after a few slots. Note that in order to meet the AoI constraint, the sensors may be enforced to transmit their sample even if the power consumption is excessive.
%
%Then, to meet the AoI constraint the sensor transmits its packet on the sub-channel even if it needs a lot of power.
 On the other hand, increasing the number of sub-channels from $N=10$ to $N=12$ yields only negligible gain. This is because the greedy sub-channel assignment policy guarantees that each sensor will be assigned at least one sub-channel.

%
%
%
% when the number of available sub-channels is less than the number of sensors, the effect of the number of sub-channels on the average total transmit  power is more than that when the number of sub-channels is greater than the number of sensors.
\begin{figure}
	\centering
	\includegraphics[scale=0.47]{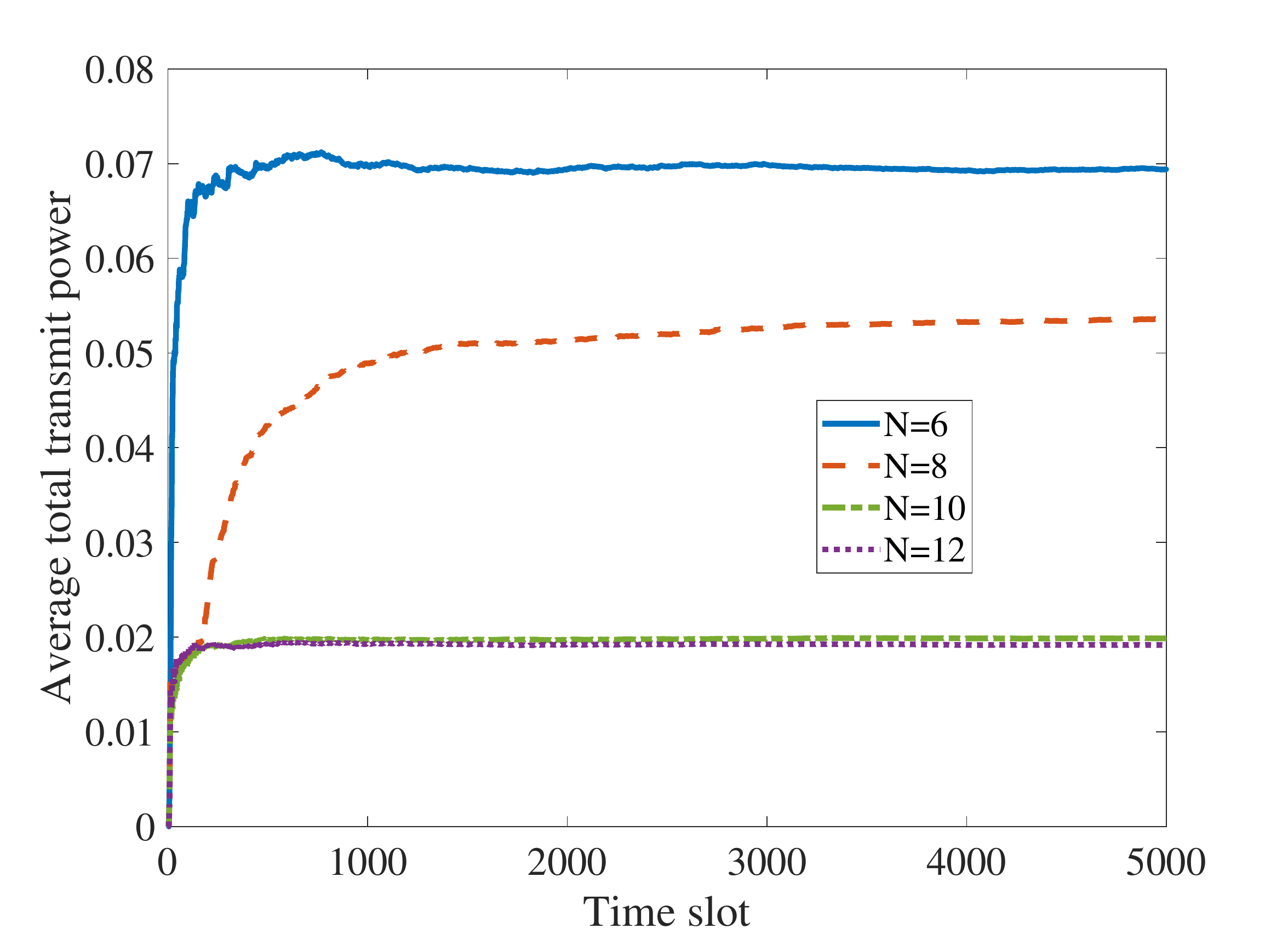}
	\caption{ Evolution of the average total transmit power of the sensors for different numbers of sub-channels $N$ with ${\Delta^{\text{max}}=4}$ and $V=8000$.}
	\vspace{-10mm}
	\label{lPower101520v800}
\end{figure}

        Fig. \ref{lPowerDeltamax234V800} illustrates the evolution of the average total transmit power for different values of maximum AoI ${\Delta^{\text{max}}}$ with  $N=10$ sub-channels and $V=8000$. The figure shows that when  $\Delta^{\text{max}}$ decreases, the average total transmit power increases. This is because  when $\Delta^{\text{max}}$ decreases, each sensor needs to take samples more frequently to satisfy constraint \eqref{eqo5}. 
        % In addition, we can see that when ${\Delta^{\text{max}}}$ increases the average total transmit power starts to saturate  into a certain level.

         \begin{figure}
      	\centering
      	\includegraphics[scale=0.47]{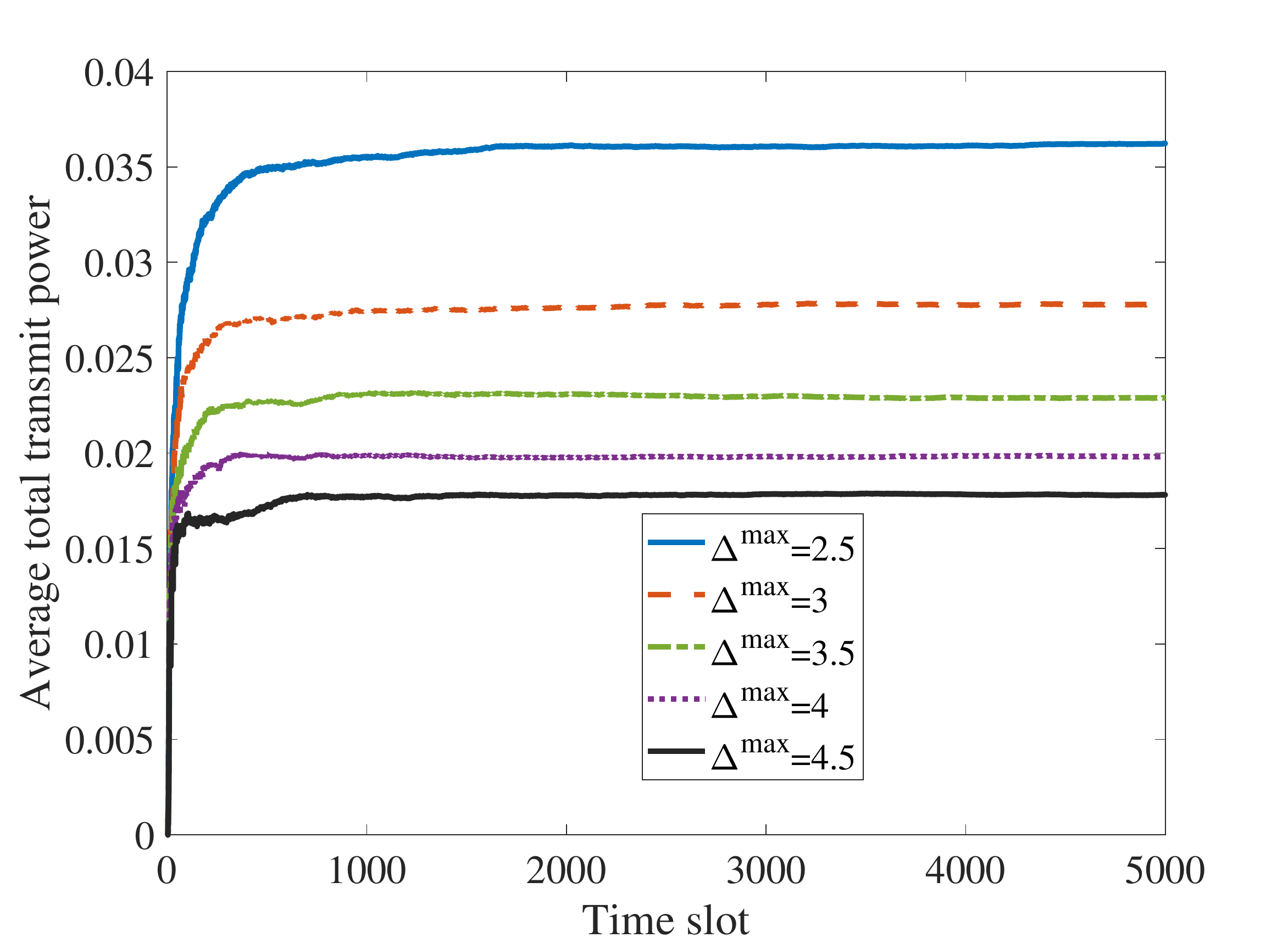}
      	\caption{  Evolution of the average total transmit power of the sensors for different values of $\Delta^{\text{max}}$ with $N=10$ and $V=8000$. }
      		\vspace{-9mm}
      	\label{lPowerDeltamax234V800}
      \end{figure}

  Fig. \ref{lAgeof_each_sensor} depicts the average AoI for individual sensors as a function of $V$ for   ${\Delta^{\text{max}}=4}$ and $N=10$ sub-channels.  According to this figure, when  $V$ increases, the average AoI of each sensor increases as well. This is because  when $V$ increases, the backlogs of the virtual queues associated with the average AoI constraint \eqref{eqo5}  increase. 
  We can also observe that the average AoI of each sensor is always smaller than  the maximum acceptable average AoI ${\Delta^{\text{max}}}$.  This validates that the drift-plus-penalty method is able to meet the average constraint through enforcing the virtual queue stability.  Moreover, we can see that a sensor that has a longer distance  to the sink has higher average AoI. This is because a sensor far away from the sink must compensate for the large-scale fading by using more power, and thus, it samples rarely.   

          \begin{figure}
        	\centering
        	\includegraphics[scale=0.47]{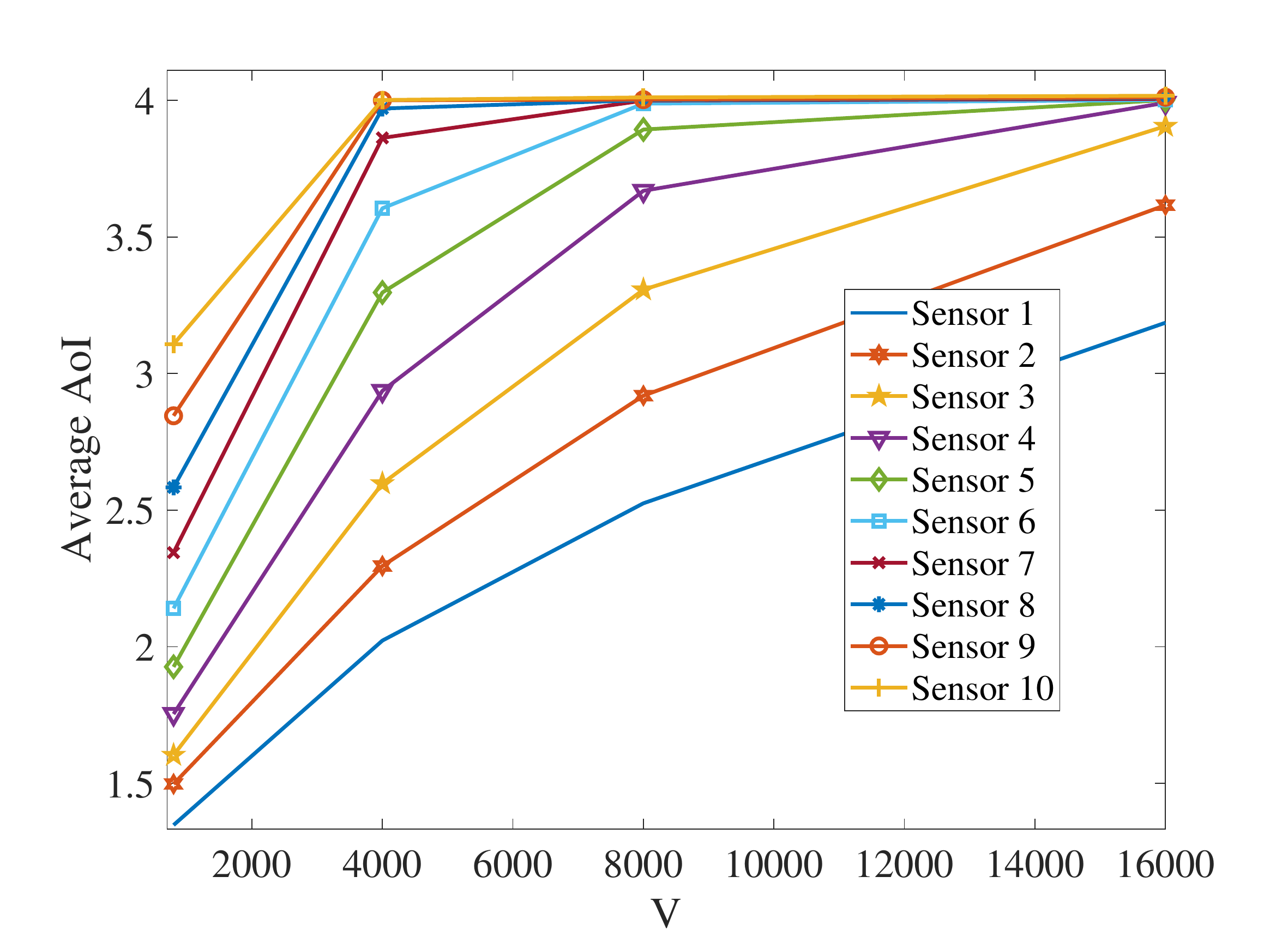}
        	\caption{  The average AoI of different sensors as a function of $V$ for ${\Delta^{\text{max}}=4}$ and $N=10$. }
        		\vspace{-8mm}
        	\label{lAgeof_each_sensor}
        \end{figure}

        \subsection{Performance of the Sub-Optimal Solution }
     To evaluate the optimality gap of the proposed sub-optimal solution for \eqref{eqoi1} presented in Algorithm \ref{table-2_1}, we compare the results obtained by  the  sub-optimal solution to those of the optimal solution calculated by the full exhaustive search method. In this regard, we consider a small setup with $K=5$ sensors $\{S_1,\ldots,S_5\}$ (see Fig. \ref{Location}) and ${N=5}$ sub-channels. 
    % The  coordinate of sensor $1$ is $(0,300)$, the coordinate of sensor $2$ is  $(300,0)$, and the coordinate of the sink is $(0,0)$.  
     The maximum acceptable average AoI of sensors is ${\Delta^{\text{max}}=4}$. 
     %To find the optimal solution, we use .
     
     Fig. \ref{f03} illustrates the evolution of the average total transmit power for different values of $V$. Fig. \ref{f01} illustrates the trade-off between the average total transmit power of the sensors and average backlogs of the virtual queues as a function of $V$. Fig. \ref{lAgeof_each_sensor2}  depicts the average AoI of different sensors as a function of $V$. Fig. \ref{f02} depicts the evolution of the average AoI of different sensors for $V=8000$.  From these figures, we can see that the proposed sub-optimal solution provides a near-optimal solution for the optimization problem \eqref{eqoi1}.

     \begin{figure}[h]
   	\centering
   	\includegraphics[scale=0.47]{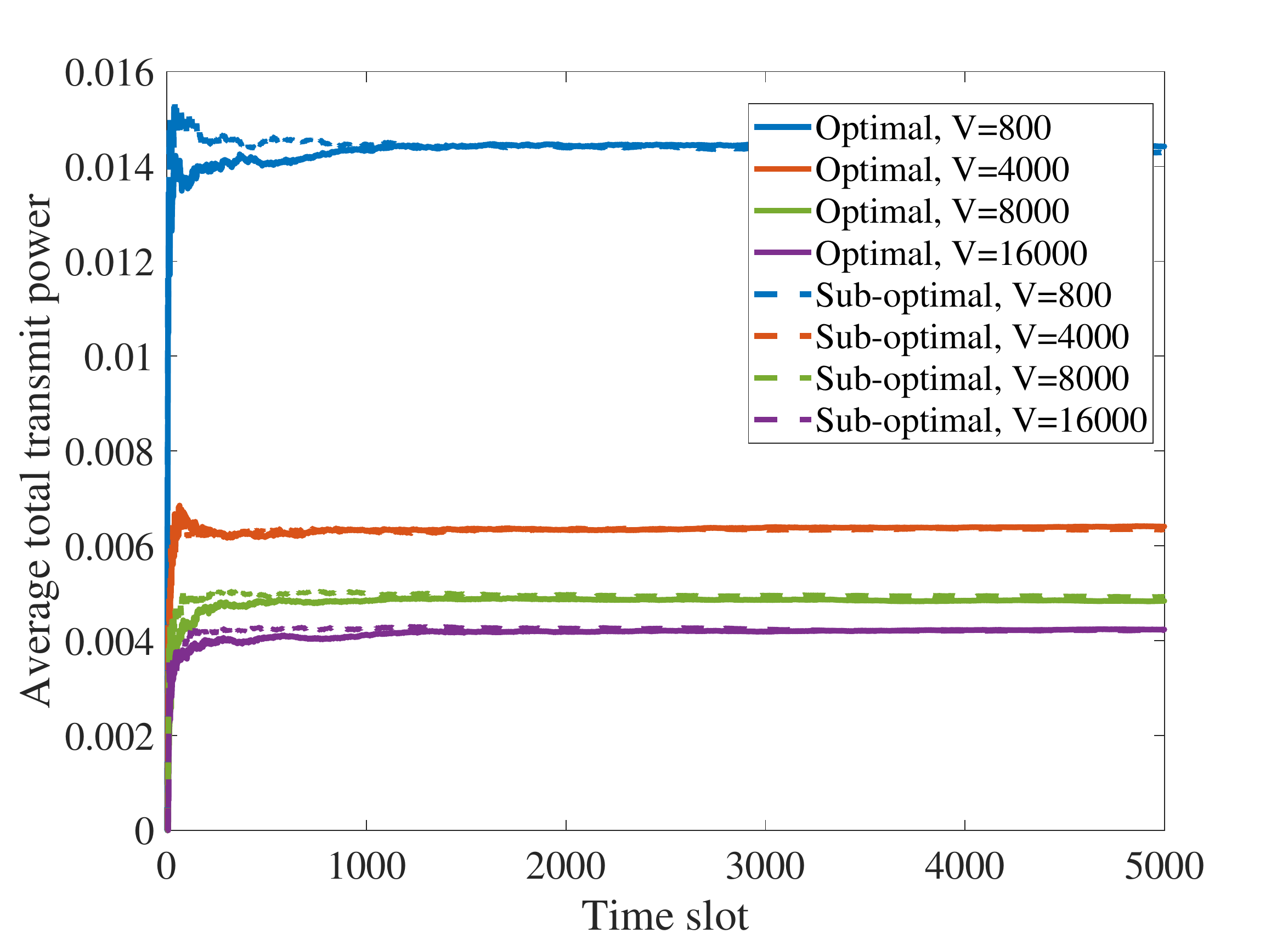}
   	\caption{ Evolution of the average total transmit power of the sensors for different values  of $V$. }
   	\vspace{-10mm}
   	\label{f03}
   \end{figure}

\begin{figure}[h]
	\centering
	\includegraphics[scale=0.47]{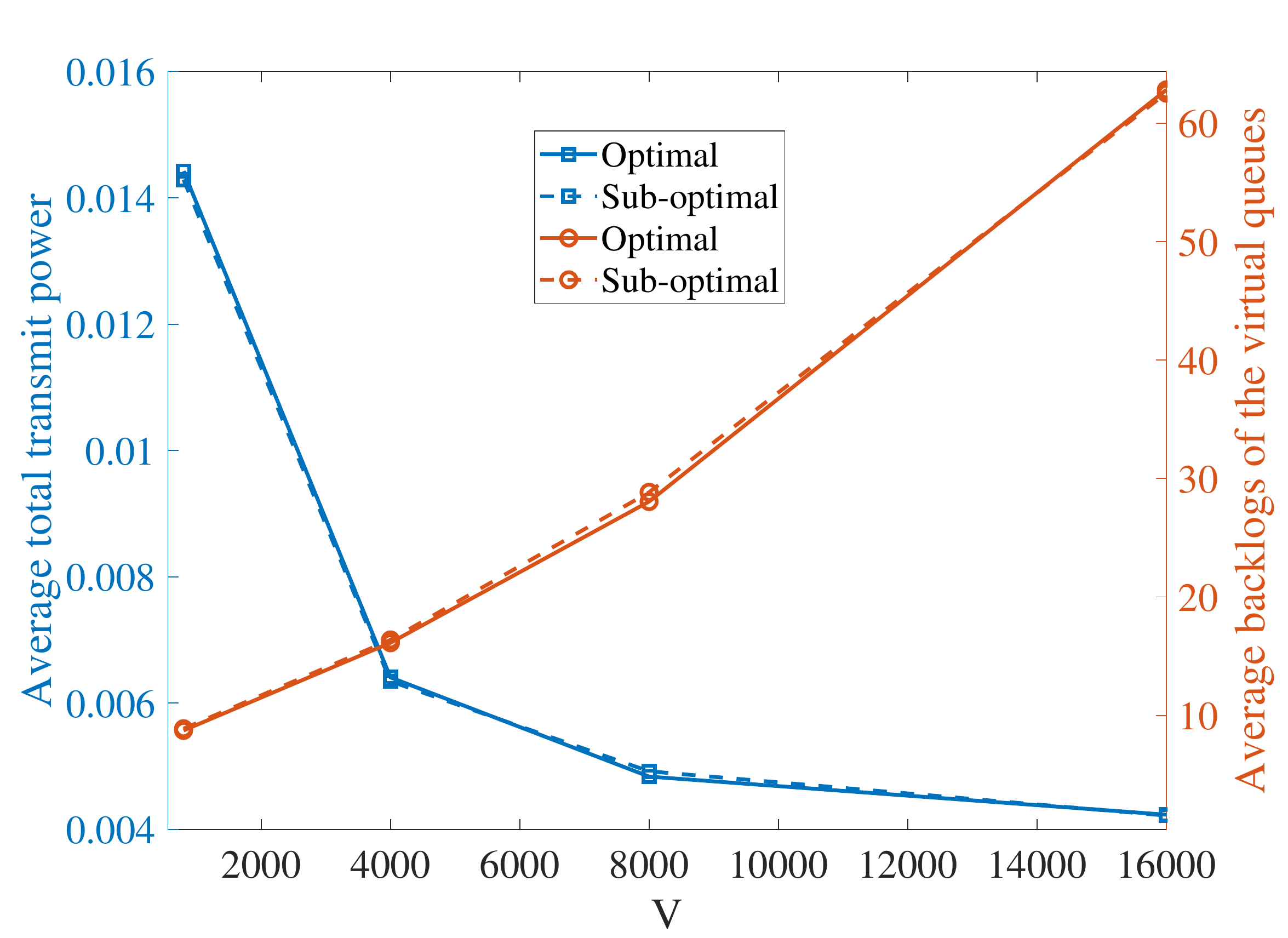}
	\caption{Trade-off between the average total transmit power and average backlogs of the virtual queues as a function of $V$. }
	\vspace{-9mm}
	\label{f01}
\end{figure}
\begin{figure}[h]
	\centering
	\includegraphics[scale=0.47]{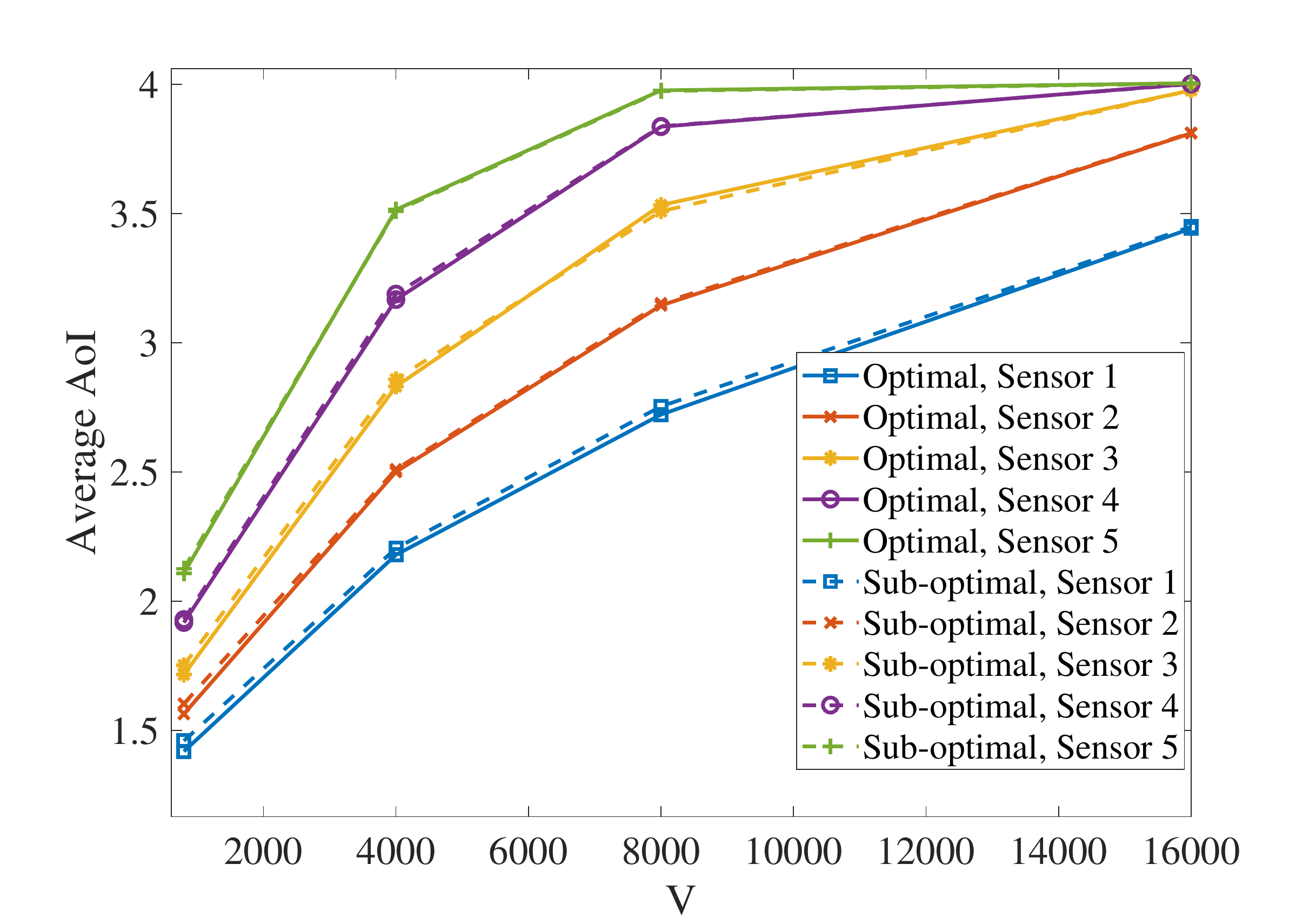}
	\caption{  The average AoI of different sensors as a function of $V$. }
		\vspace{-9mm}
	\label{lAgeof_each_sensor2}
\end{figure}

   \begin{figure}[h]
	\centering
	\includegraphics[scale=0.47]{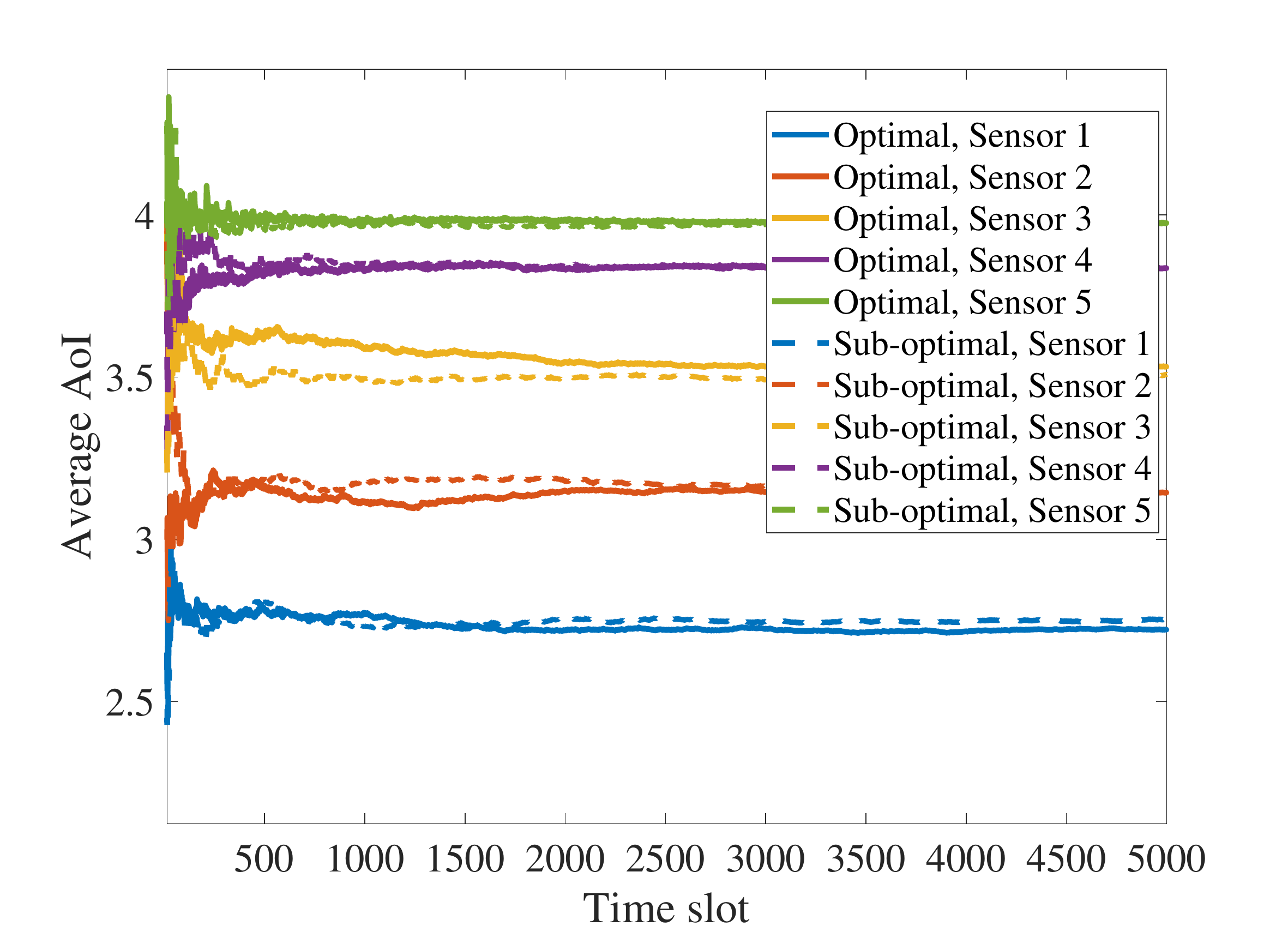}
	\caption{ Evolution of the average AoI of different sensors for $V=8000$. }
		 		\vspace{-9mm}
	\label{f02}
  \end{figure}

% \begin{figure}
%	\centering
%	\includegraphics[scale=0.4]{ageu2vessusV.eps}
%	\caption{ Time average AoI of sensor 2 as a function of $V$. }
%	\vspace{-5mm}
%	\label{f04}
%\end{figure}

%\begin{figure}
%	\centering
%	\includegraphics[scale=0.55]{AveraNetAoI800V5.eps}
%	\caption{Average network AoI  as a function of $T_{\text{max}}$. }
%		 		\vspace{-5mm}
%	\label{f05}
%\end{figure}
%  \begin{figure}
%	\centering
%	\includegraphics[scale=0.56]{Netqueueevolu.eps}
%	\caption{Evolution of the network backlog. }
%		 		\vspace{-5mm}
%	\label{f04}
%    \end{figure}
\section{Conclusions}\label{Conclusion}
We considered a status update system consisting of a set of sensors and one sink.  The sink controls the sampling process of the sensors in a way that it decides whether each sensor takes a sample or not
at the beginning of each slot. The status update packets of the sensors are transmitted by sharing a set of orthogonal sub-channels in each slot.  We formulated a problem to minimize the average total transmit power of sensors under the average AoI  constraint for each sensor. To solve the proposed problem, we used the Lyapunov drift-plus-penalty method. This method provides a trade-off between the average total transmit power and the average AoI of the sensors. We conducted optimality analysis to study this  trade-off. In the numerical results section,  we  showed the performance of the proposed  dynamic solution algorithm in terms of transmit power consumption and AoI of sensors. The results showed that by using the proposed dynamic control algorithm more than $60~\%$ saving in the average total transmit power  can be achieved compared to a baseline policy.  In addition, we showed that the sub-optimal solution for the per-slot optimization problems provides a near-optimal solution. 

%\textcolor{blue}{From the results we can conclude that by using  $V$ (as the Lyapunov drift-plus-penalty method parameter) we can provide a trade-off between the  average backlogs of the virtual queues and the average total transmit power in the system.  It is worth noting that $ V $ cannot has an infinite value because it violates the stability of the virtual queues which means the time average constraints can not be satisfied. Moreover, the results demonstrated that when $ V $ is sufficiently large, increasing it further does not significantly reduce the power.   We also can conclude  that while we can provide a trade-off between the  average backlogs of the virtual queues and the average total transmit power in the system, the average AoI of each sensor is always smaller than  the maximum acceptable average AoI.  This validates that the drift-plus-penalty method is able to meet the average constraint through enforcing the virtual queue stability.}
The numerical results illustrated the inherent trade-off between the average AoI of the sensors and the
	average total transmit power that the Lyapunov drift-plus-penalty method brings in the system. 
	This trade-off is adjusted by the penalty parameter $V$. 	
	A high value of $V$ is beneficial in that it enforces smaller transmit powers, yet at the cost of increasing the average AoI of each sensor. 	
	The results validated that, regardless of the value of $V$, the proposed drift-plus-penalty method met the time average AoI constraints through successfully enforcing the virtual queue stability.	
	Regarding a selection of parameter $V$ in practice, we observed that when $V$ is sufficiently large, increasing it further does not significantly reduce the power.	
%	In any case, infinite $V$ would violate the stability criterion for the virtual queues. 

%  \section*{Acknowledgements}
%  This research has been financially supported by the Infotech Oulu, the Academy of Finland (grant 323698), and Academy of Finland 6Genesis Flagship (grant 318927). M. Codreanu would like to acknowledge the support of the European Union's Horizon 2020 research and innovation programme under the Marie Sk\l{}odowska-Curie Grant Agreement No. 793402 (COMPRESS NETS). The work of M. Leinonen has also been financially supported in part by the Academy of Finland (grant 319485). M. Moltafet would like to acknowledge the support of Finnish Foundation for Technology Promotion and  HPY Research Foundation.

%  \appendices
%  \section{Proof of Lemma \ref{lem1onedrift}} \label{Valuesof v appendix}

\bibliographystyle{IEEEtran}
%% \mybibliography
\bibliography{Bibliography}

\end{document}